%% file: main.tex
\documentclass{article}[11pt, margin=1in]
\usepackage[T1]{fontenc}
\usepackage{a4wide}
\usepackage{graphicx}
\usepackage{tikz}
\usepackage{amsmath}
\usepackage{amsthm}
\usepackage{thmtools}
\usepackage{amssymb}
\usepackage[most]{tcolorbox}
\usepackage{thm-restate}
\usetikzlibrary{fit,shapes.geometric,backgrounds}
\usepackage{complexity}
\usepackage{authblk}

\usepackage{algorithm}
\PassOptionsToPackage{noend}{algpseudocode}
\usepackage{algpseudocode}
\usepackage[pdftex,colorlinks,linkcolor=black,urlcolor=black,citecolor=black]{hyperref}
\usepackage[noabbrev,capitalise,nameinlink]{cleveref}

\makeatletter
\@addtoreset{ALG@line}{algorithm}
\makeatother

\makeatother
\ifpdf
\fi

\newtheorem{theorem}{Theorem}

\newtheorem{lemma}[theorem]{Lemma}

\newtheorem*{claim*}{Claim}

\theoremstyle{remark}
\newtheorem{remark}[theorem]{Remark}

\newcommand*\samethanks[1][\value{footnote}]{\footnotemark[#1]}

\date{}

\begin{document}

\title{A Polynomial-Time Approximation Algorithm for Complete Interval Minors}
\author[1]{Romain Bourneuf\thanks{Email: \texttt{\{romain.bourneuf, julien.cocquet, chaoliang.tang, stephan.thomasse\}@ens-lyon.fr}}}
\affil[1]{Univ. Bordeaux, CNRS, Bordeaux INP, LaBRI, UMR 5800, F-33400 Talence, France.}
\author[2]{Julien Cocquet\samethanks}
\author[2,3]{Chaoliang Tang\samethanks}
\author[2]{Stéphan Thomassé\samethanks}
\affil[2]{ENS de Lyon, Université Claude Bernard Lyon 1, CNRS, Inria, LIP UMR 5668, Lyon, France.}
\affil[3]{Shanghai Center for Mathematical Statistics, Fudan University, 220 Handan Road, Shanghai 200433, China}

\maketitle

\begin{abstract}
As shown by Robertson and Seymour, deciding whether the complete graph $K_t$ is a minor of an input graph $G$ is a fixed parameter tractable problem when parameterized by $t$. From the approximation viewpoint, the gap to fill is quite large, as there is no PTAS for finding the largest complete minor unless $\P = \NP$, whereas a polytime $O(\sqrt n)$-approximation algorithm was given by Alon, Lingas and Wahlén.

We investigate the complexity of finding $K_t$ as interval minor in ordered graphs (i.e. graphs with a linear order on the vertices, in which intervals are contracted to form minors). Our main result is a polytime $f(t)$-approximation algorithm, where $f$ is triply exponential in $t$ but independent of $n$. The algorithm is based on delayed decompositions and shows that ordered graphs without a $K_t$ interval minor can be constructed via a bounded number of three operations: closure under substitutions, edge union, and concatenation of a stable set. As a byproduct, graphs avoiding $K_t$ as an interval minor have bounded chromatic number.
\end{abstract}

\input{introduction}
\input{delayed_decomp}
\input{delayed_rank}
\input{Analysis_algo}

\section*{Acknowledgements}

We thank Julien Duron for his help with the tedious calculations.

The authors are partially supported by the French National Research Agency under research grants ANR GODASse ANR-24-CE48-4377 and ANR Twin-width ANR-21-CE48-0014-01. The third authors is also supported by Chinese Scholarship Council (CSC).

\bibliographystyle{plain}
\bibliography{biblio}

\appendix
\input{appendix}

\end{document}

%% file: introduction.tex
\section{Introduction}

Complete minors in graphs form an extensively studied subject, notably featuring the fundamental result of Robertson and Seymour~\cite{RS95}, which asserts that testing whether the complete graph $K_t$ is a minor of an input graph $G$ on $n$ vertices can be done in time $f(t) \cdot n^3$. 
This was proved in the series of Graph Minors papers, which also provided a decomposition theorem of $K_t$-minor-free graphs~\cite{RS03} (whose bounds were recently improved significantly by Gorsky, Seweryn and Wiederrecht~\cite{GSW25}). 
One of the key basic facts allowing such a result is that $K_t$-minor-free graphs are sparse, i.e. they have a linear number of edges. 
Recently, Korhonen, Pilipczuk and Stamoulis~\cite{KPS24} provided an algorithm running in time $f(t) \cdot n^{1+o(1)}$ for the same problem, improving on a quadratic algorithm from Kawarabayashi, Kobayashi and Reed~\cite{KKR12}. 
From the approximation point of view, the landscape is much less understood. 
The first hardness result is due to Eppstein~\cite{Epp09}: finding the size of a largest complete minor is $\NP$-hard. 
This was extended by Wahlén~\cite{Wah09}, who showed that there is no polynomial time approximation scheme (PTAS) for the size of a largest complete minor, unless $\P = \NP$. 
The current best known approximation factor achievable in polynomial time is $O(\sqrt n)$, as shown by Alon, Lingas and Wahlén~\cite{ALW07}. 
Up to this day, it is still open whether this problem admits a polytime $f(OPT)$ approximation algorithm. The goal of this paper is to investigate the corresponding problem for complete interval minors in ordered graphs. 

An \emph{ordered graph} $(G,<)$ is a graph $G=(V,E)$ equipped with a linear order $<$ on its vertices. We will consistently denote the number of vertices by $n$, and the number of edges by $m$. We usually denote the vertices as $v_1,\dots ,v_n$, enumerated according to $<$.
For simplicity of notation, we will sometimes omit the order $<$ and talk about an ordered graph $G$.
An \emph{interval minor} of $(G,<)$ is obtained by deleting some edges of $G$ as well as iteratively contracting pairs of vertices which are consecutive in $<$. An example is displayed in \cref{fig:ex-im}.

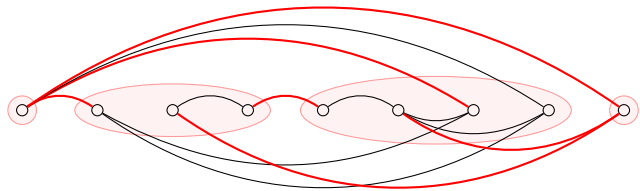
\begin{figure}[ht]
	\centering
	\begin{tikzpicture}[
		every node/.style={circle,draw,inner sep=1.5pt},
		arc/.style={bend left=35} 
		]

		\foreach \i in {1,...,9} {
			\node (v\i) at (\i,0) {};
		}

		\begin{scope}[on background layer]
			\node[draw=red!40, fill=red!5, fit=(v1), ellipse, inner sep=1.5pt] {};
			\node[draw=red!40, fill=red!5, fit=(v9), ellipse, inner sep=1.5pt] {};
			\node[draw=red!40, fill=red!5, ellipse, minimum width=2.6cm, minimum height=0.7cm] at (3,0) {};
			\node[draw=red!40, fill=red!5, ellipse, minimum width=3.6cm, minimum height=0.9cm] at (6.5,0){};
		\end{scope}

		\draw[thick, red, arc] (v1) to (v2);
		\draw[thick, red, arc, bend left=32] (v1) to (v7);
		\draw[arc, bend left=33] (v1) to (v8);
		\draw[thick, red,arc] (v1) to (v9);
		
		\draw[arc, bend left=29] (v7) to (v2);
		\draw[arc] (v8) to (v2);
		\draw[thick, red, arc] (v9) to (v3);
		\draw[arc, bend left=25] (v7) to (v6);
		\draw[arc, bend left=30] (v8) to (v6);
		\draw[thick, red, arc,] (v9) to (v6);
		
		\draw[arc] (v3) to (v4);
		\draw[thick, red, arc] (v4) to (v5);
		\draw[arc] (v5) to (v6);
	\end{tikzpicture}
	\caption{A $K_4$ interval minor: the red zones represent the intervals we contract and the red edges are the remaining edges. Observe that this interval minor model is not a minor model since we contracted non-connected subsets of vertices.}
	\label{fig:ex-im}
\end{figure}

The central computational problem on interval minors is the following:

\begin{tcolorbox}[colback=white,colframe=black]
\textsc{Interval Minor Detection} \\
\textbf{Input:} Two ordered graphs $G$ and $H$. \\
\textbf{Question:} Is $H$ an interval minor of $G$?
\end{tcolorbox}

The complete study of the computational complexity of this question is probably very challenging, and a good first step is to limit ourselves (as for usual minors) to the case where $H$ is a complete graph. 
The crucial fact to observe is that $K_t$-interval-minor-free ordered graphs can be very complex. Say that a bipartite ordered graph $(G, <)$ with edges between two parts $X$ and $Y$ is \emph{monotone} if either $X < Y$ or $Y < X$. Then, observe that the monotone $K_{t,t}$ does not contain $K_4$ as an interval minor. 
In particular, $K_4$-interval-minor-free ordered graphs can have a quadratic number of edges. 
They also form a large class of graphs (with growth at least $2^{n^2/4}$) as they contain all subgraphs of the monotone $K_{n/2,n/2}$. 
Therefore, any attempt to effectively construct $K_t$-interval-minor-free ordered graphs must involve an operation allowing the creation of arbitrary monotone bipartite graphs. 
The landscape of $K_t$-interval-minor-free ordered graphs thus seems very different from the one of $K_t$-minor-free graphs. 
However, a strong connection exists between these two worlds: the right analogy with graph minors involves monotone-$K_{t,t}$-interval-minor-free ordered graphs.
To see this, let us first recall the celebrated Marcus-Tardos theorem~\cite{MT04}: 

\begin{theorem}\label{thm:marcustardos}
There exists a function $f$ such that every ordered graph on $n$ vertices with at least $f(t) \cdot n$ edges contains a monotone $K_{t,t}$ as an interval minor.
\end{theorem}

The notion of interval minors was formally introduced by Fox~\cite{fox2013} to prove bounds on the function $f$ of \cref{thm:marcustardos}. The upper bound was later improved by Cibulka and Kynčl~\cite{CK17}, again using interval minors.
\cref{thm:marcustardos} gave a positive answer to a conjecture of Füredi and Hajnal~\cite{FH92}, stating that for every permutation matrix $P$, there exists a constant $c_P$ such that every $n \times n$ binary matrix with at least $c_P \cdot n$ entries 1 contains $P$ as a pattern. 
This result was later generalized to matrices in higher dimension by Klazar and Marcus~\cite{KM07}. 
The corresponding bound was later improved by Geneson and Tian~\cite{GT17}, once more using interval minors.
In~\cite{fox2013}, Fox suggested to conduct a thorough study of ordered graphs avoiding a given interval minor, with the hope of developing a theory analogue to the graph minor theory of Robertson and Seymour. 
There have been several papers going in that direction, see for instance~\cite{MRTW16,MLWG15,JK18}.
In this paper, we continue this line of research by providing a decomposition theorem for $K_t$-interval-minor-free graphs.

The Marcus-Tardos theorem therefore implies that monotone-$K_{t,t}$-interval-minor-free ordered graphs are sparse, as are $K_t$-minor-free graphs (in the non-ordered case). A second analogy between these two classes appears from the computational angle:

\begin{theorem}\label{thm:Kttfpt}
There exists a function $g$ such that testing whether an ordered graph $(G,<)$ contains a monotone $K_{t,t}$ as an interval minor can be done in time $g(t) \cdot |V(G)|$.
\end{theorem}

The cornerstone of this algorithm is the $\FPT$ algorithm of Guillemot and Marx~\cite{GM14} for detecting a pattern in a permutation. 
Their algorithm amounts to detecting a monotone $K_{t,t}$ as an interval minor in a monotone matching. 
The generalization to arbitrary ordered graphs follows from the notion of twin-width introduced in \cite{BKTW21}.
More specifically, approximating the twin-width of an ordered graph can be done in $\FPT$ time, see \cite{BGOSTT24}. 
In a nutshell, if an ordered graph $(G, <)$ does not contain a monotone complete bipartite graph as an interval minor, then its twin-width is bounded, and therefore dynamic programming can test any first-order formula in $\FPT$ time. 
This gives a win/win algorithm: if the twin-width is large compared to $t$, simply return Yes since $(G, <)$ must contain $K_{t,t}$ as interval minor, and if not, one can test if $K_{t,t}$ is an interval minor of $(G, <)$, as first-order logic can encode interval minors. 
Indeed, $(G, <)$ contains an ordered graph $(H, <')$ on vertex set $u_1 <' \ldots <' u_h$ if and only if there exist vertices $x_1, \ldots, x_{h-1} \in V(G)$ such that $x_1 < \ldots < x_{h-1}$, and for every edge $u_iu_j \in E(H)$, there exists $y_i, y_j \in V(G)$ such that $x_{i-1} < y_i \leq x_i$, $x_{j-1} < y_j \leq x_j$ and $y_iy_j \in E(G)$ (with the obvious adaptation for $u_1$ and $u_h$). 
Let us end this detour about $K_{t,t}$ interval minors with the central open question in this topic: 
Given a graph $G$ which admits an order $<$ avoiding $K_{t,t}$ as an interval minor, can we efficiently compute an order $<'$ such that $(G,<')$ avoids $K_{f(t),f(t)}$ as an interval minor. 
This question is equivalent to ask for a $f(OPT)$-approximation algorithm of the twin-width for sparse graphs (see Theorem 2.12 in~\cite{BGKTW22}).

Let us go back to our original goal of approximately detecting $K_t$ as an interval minor. 
We follow the classical strategy of providing an effective decomposition of $K_t$-interval-minor-free ordered graphs using a bounded number of tractable operations. 
The key-tool for this is the notion of \emph{delayed decomposition}, used in~\cite{PS23} and formally defined in~\cite{BT25} to show that graphs of bounded twin-width are polynomially $\chi$-bounded. 
It was later used as the central tool to show that pattern-avoiding permutations are the product of a bounded number of separable permutations~\cite{BBGT24}. 
We first introduce three operations on classes of ordered graphs. Given two classes $\mathcal C, \mathcal{C}'$ of ordered graphs, we define:

\begin{itemize}
    \item The \emph{substitution closure} of $\mathcal{C}$ is the class $\overline {\mathcal C}$ of ordered graphs which can be obtained from $\mathcal C$ by iterating an arbitrary number of substitutions of elements of $\mathcal C$. The \emph{substitution} of a vertex $v$ of an ordered graph $(G,<)$ by an ordered graph $(H,<')$ consists of replacing $v$ by a copy of $(H,<')$ and joining all neighbors of $v$ to all vertices of $H$. The vertices of $H$ remain ordered according to $<'$, at the position of $v$ in $<$.
    
    \item The \emph{edge union} of $\mathcal C$ and $\mathcal{C}'$ is the class $\mathcal C \oplus \mathcal{C}'$ of all ordered graphs of the form $(G \oplus G',<)$ where $(G,<)\in {\mathcal C}$ and $(G',<)\in {\mathcal C'}$ are two graphs on the same set of vertices ordered by the same linear order $<$, and $G \oplus G'$ is their edge union. 
    \item The \emph{independent concatenation} of $\mathcal{C}$ is the class $\mathcal C^+$ of all ordered graphs $(G^+,<^+)$ which can be obtained from a graph $(G,<) \in \mathcal C$ by adding an independent set $I$ at the beginning or at the end of $<$, which can be connected arbitrarily to the vertices of $G$. Formally, either $V(G)<^+I$ or $I<^+V(G)$.
\end{itemize}

We define the \emph{rank} of an ordered graph $(G, <)$ inductively, as follows:
\begin{itemize}
    \item The empty graph has rank $0$.
    \item The rank of $(G, <)$ is the smallest integer $r$ such that $(G, <)$ can be built from graphs of rank at most $r-1$, using one of the following operations: substitution closure, edge union or independent concatenation.
\end{itemize}
Let us illustrate this parameter. The class of ordered edgeless graphs has rank 1, as they can be obtained by an independent concatenation from the empty graph. The class of monotone bipartite graphs has rank 2, as we can perform another round of independent concatenation. The class of ordered complete graphs has rank at most 3, since the edge graph $K_2$ is obtained at rank 2, and its closure under substitution gives all complete graphs. A slightly more involved example is that of the ordered path $(P_n,<)$, in which each vertex is joined to its successor in $<$. 

Note that the subgraph $M_o$ of $P_n$ consisting of the odd indexed edges is a matching consisting of edges $e_1<e_3<\dots$, hence can obtained by substituting the vertices of an independent set by edges.
In particular $M_o$ has rank 3, and the even indexed edges subgraph $M_e$ also has rank 3. Finally, $P_n=M_o \oplus M_e$ has rank 4. This is illustrated in~\cref{fig:P6}.

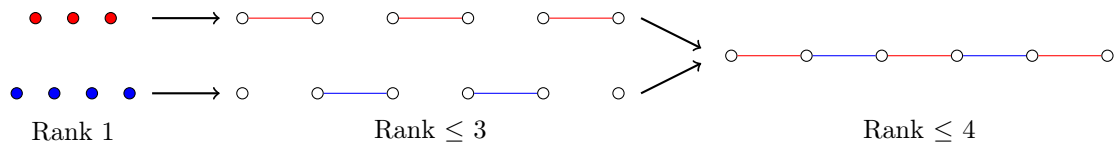
\begin{figure}[ht]
	\centering
	\begin{tikzpicture}[
		every node/.style={circle,draw,inner sep=1.5pt}
		]

		\node[fill = red] (r1) at (0.25,1) {};
		\node[fill = red] (r2) at (0.75,1) {};
		\node[fill = red] (r3) at (1.25,1) {};
		
		\node[fill =blue] (b1) at (0,0) {};
		\node[fill =blue] (b2) at (0.5,0) {};
		\node[fill =blue] (b3) at (1,0) {};
		\node[fill =blue] (b3) at (1.5,0) {};

		\draw[->,thick](1.8,1)to(2.7,1);
		\node (rr1) at (3,1) {};
		\node (rr2) at (4,1) {};
		\node (rr3) at (5,1) {};
		\node (rr4) at (6,1) {};
		\node (rr5) at (7,1) {};
		\node (rr6) at (8,1) {};
		\draw[red] (rr1) -- (rr2);
		\draw[red] (rr3) -- (rr4);
		\draw[red] (rr5) -- (rr6);
		
		\draw[->,thick](1.8,0)to(2.7,0);
		\node (bb1) at (3,0) {};
		\node (bb2) at (4,0) {};
		\node (bb3) at (5,0) {};
		\node (bb4) at (6,0) {};
		\node (bb5) at (7,0) {};
		\node (bb6) at (8,0) {};
		\draw[blue] (bb2) -- (bb3);
		\draw[blue] (bb5) -- (bb4);

		\draw[->,thick](8.3,1)to(9.1,0.6);
		\draw[->,thick](8.3,0)to(9.1,0.4);
		\node (p1) at (9.5,0.5) {};
		\node (p2) at (10.5,0.5) {};
		\node (p3) at (11.5,0.5) {};
		\node (p4) at (12.5,0.5) {};
		\node (p5) at (13.5,0.5) {};
		\node (p6) at (14.5,0.5) {};
		\draw [red](p1) -- (p2);
		\draw [blue](p2) -- (p3);
		\draw [red] (p3)-- (p4);
		\draw[blue](p4) -- (p5);
		\draw[red](p5) -- (p6);

		\node[draw = none] at(.75,-0.5) {Rank 1};
		\node[draw = none] at(5.5,-0.5) {Rank $\le$ 3};
		\node[draw = none] at(12,-0.5) {Rank $\le$ 4};
	\end{tikzpicture}
	\caption{Construction showing that the ordered $P_6$ has rank at most $4$. The edge graph $K_2$ has rank $2$, which is why we go from rank $1$ to rank $\le3$.}
	\label{fig:P6}
\end{figure}

The central result of this paper is the following:

\begin{restatable}{theorem}{bdrk}\label{thm:boundedrank}
Every $K_t$-interval-minor-free ordered graph $(G, <)$ has bounded rank.
\end{restatable}

The cornerstone of \cref{thm:boundedrank} is the interplay between the two operations of substitution closure and edge union (as is the case for the construction of a path from two matchings). 
Before introducing the notion of delayed decomposition, let us first formalize what it means exactly that a graph $G$ belongs to the substitution closure of a class $\mathcal C$. 
The adequate representation of $G$ must take into account that some vertices have been repeatedly substituted by a graph of $\mathcal C$, hence $G$ can be expressed as a tree of substitutions. 
Formally, a \emph{$\mathcal C$-substitution tree} of $G$ is a rooted tree $T$ whose leaves are the vertices of $G$. 
Moreover, every internal node $x$ of $T$ is labeled by a graph $G_x$ of $\mathcal C$, whose vertices are the children of $x$ in $T$. 
Finally, two vertices $u, v$ of $G$ form an edge if and only if given their closest common ancestor $x$, the two children $u',v'$ of $x$ which are the respective ancestors of $u, v$ form an edge in $G_x$.
By construction, the graphs in the substitution closure of $\mathcal C$ are precisely the ones representable by a $\mathcal C$-substitution tree. 
These structured trees were introduced by Gallai to study partial orders~\cite{Gallai67}.
The most popular examples of substitution trees are those used to decompose cographs ($P_4$-free graphs) using binary trees in which the graphs $G_x$ are the edge and the non edge. 
These trees are well-suited for ordered graphs, as the order $<$ can be represented as the left-to-right order on the leaves.

However, most graphs $G$ do not admit a non-trivial substitution tree, i.e. one in which the root is not labeled by $G$. 
For instance paths of length at least 3 are indecomposable (or prime) with respect to substitutions.
A simple way to strengthen substitution trees is to delay the effect of the graphs $G_x$: instead of creating edges between children, they act on their grandchildren. 
Let us formalize this:

A \emph{delayed structured tree} is a rooted tree $T$ whose leaves are the vertices of a graph $G$ (the \emph{realization} of $T$).
Moreover, every internal node $x$ of $T$ is labeled by a \emph{quotient graph} $G_x$ whose vertices are the grandchildren of $x$ in $T$. 
Finally, two vertices $u,v$ of $G$ are connected if and only if given their closest common ancestor $x$, the two grandchildren $u',v'$ of $x$ which are the respective ancestors of $u,v$ are connected in $G_x$. For an example, see \cref{fig:delayed}. 
Authorizing this delay results in a tool which is much more expressive than substitution trees. 
In particular, given a class of graphs $\mathcal{C}$, the class $\mathcal C'$ of realizations of delayed structured trees whose quotient graphs $G_x$ belong to $\mathcal C$ is much harder to grasp than the mere substitution closure of $\mathcal C$.
 
However, the class $\mathcal C'$ is not too complex compared to $\mathcal C$: given the delayed structured tree $T$, consider the quotient graphs $G_x$ of the nodes $x$ with even depth, and of the ones with odd depth.
Now form two trees $T_e$ and $T_o$ from $T$, by setting in $T_e$ all quotient graphs $G_x$ at odd depth as edgeless graphs, and setting in $T_o$ all quotient graphs $G_x$ at even depth as edgeless graphs. 
The key observation is that the realization $G_e$ of $T_e$ belongs to the substitution closure of $\mathcal C$. 
The same holds for $G_o$, and thus $G$ is the edge union of two graphs in the substitution closure of $\mathcal C$.

In a nutshell, the proof of \cref{thm:boundedrank} is now straightforward: we just have to show that every ordered graph $(G,<)$ without $K_t$-interval minor is the realization of a delayed structured tree whose quotient graphs are \emph{simpler} than $G$, and that these simpler graphs are again decomposable into simpler objects, and that this process has a bounded height of recursion. This is the main technical part of the paper. 
Note that this process is too tame to create high-entropy objects such as monotone bipartite graphs. 
Here a choice has to be made: either declare that the basic class is indeed all monotone bipartite graphs, and just consider substitution closure and edge union, or set the empty graph as our basic class, and authorize independent set concatenation.
We chose the latter convention as it fits more in our algorithmic purpose, but we feel that the former is more in the spirit of classical graph decompositions, with only two tame operations, and the basic class  consisting of the indecomposable (or prime) monotone bipartite graphs. 

The proof of \cref{thm:boundedrank} is algorithmic, and a sequence of operations for constructing $(G, <)$ can be effectively computed in polynomial time by iterating delayed decompositions. 
This is the first phase of our algorithm to detect $K_t$ as an interval minor: either we fail to achieve bounded rank and then find $K_t$ as an interval minor, or we compute a sequence of operations achieving bounded rank for $(G, <)$. 
Unfortunately, a graph can have bounded rank and still contain a $K_t$ interval minor. 
Indeed, as we saw before, complete graphs have rank at most 3. The nice point is that large complete subgraphs are the only reason why an ordered graph of bounded rank can contain arbitrarily large complete interval minors. The second phase of the algorithm uses the decomposition of $(G, <)$ and either output a $K_t$ subgraph, or certifies that $G$ has no $K_{f(t)}$ interval minor. As a result, we show:

\begin{restatable}{theorem}{thealgo} \label{thm:the-algo}
There is a triply exponential function $f$ and a decision algorithm which, given as input an ordered graph $(G,<)$ with $n$ vertices and $m$ edges and an integer $t$, satisfies the following:\begin{itemize}
    \item If the algorithm returns Yes then $(G, <)$ contains $K_t$ as an interval minor.
    \item If the algorithm returns No then $(G, <)$ does not contain  $K_{f(t)}$ as an interval minor.
    \item The algorithm runs in time $O(t \cdot mn^2)$.
\end{itemize}
\end{restatable}

Before diving into the details, we now discuss two hardness results in the context of $K_t$-minors.

\subsection*{Hardness of ordering a graph}

A natural question regarding interval minors in ordered graphs is to find the minimum $t$, such that an input graph $G$ admits an ordering $<$ with no $K_t$ interval minor. Denote this value by $kim(G)$. This question is particularly interesting since the same problem for $K_{t,t}$ instead of $K_t$ amounts to approximating the twin-width of a sparse graph. Unfortunately the answer for $K_t$ is as hard as approximating the chromatic number $\chi$ of a graph:

\begin{lemma}\label{lem:orderable}
The parameters $kim(G)$ and $\chi(G)$ are functionally equivalent.
\end{lemma}

\begin{proof}
If a graph has chromatic number $t$, ordering the vertices according to the color classes directly gives an order without $K_{2t}$ interval minor. Conversely, if a graph $G$ has an ordering $<$ such that $(G,<)$ does not contain a $K_t$ interval minor, \cref{thm:boundedrank} implies that $(G, <)$ has bounded rank.
Say that a class $\mathcal{C}$ of graphs is \emph{$\chi$-bounded} if there exists a function $h$ such that every graph $G$ in ${\mathcal C}$ satisfies $\chi(G)\leq h(\omega (G))$, where $\omega (G)$ is the maximum size of a clique of $G$. We speak of \emph{polynomial} $\chi$-boundedness if $h$ is a polynomial.

\begin{claim*}
For every $r \geq 0$, the class ${\mathcal C}_r$ of ordered graphs with rank at most $r$ is polynomially $\chi$-bounded.
\end{claim*}

\begin{proof}[\textit{Proof of the Claim}]
We prove it by induction on $r$. The statement is trivial for $r = 0$ since the empty graph has clique number $0$ and chromatic number $0$.
For the induction step, we just have to show that the three operations preserve polynomial $\chi$-boundedness.
Independent concatenation increases the chromatic number by at most one and cannot decrease the clique number. 
Similarly, the chromatic number of the edge union of two graphs is at most the product of their chromatic numbers, and the clique number of the edge union of two graphs is at least the maximum of the two clique numbers.
The fact that the substitution closure preserves polynomial $\chi$-boundedness was shown by Chudnovsky, Penev, Scott and Trotignon~\cite{CPST13}.
\end{proof}

Then, if $(G, <)$ does not contain a $K_t$ interval minor then $G$ does not contain a clique of size $t$, and therefore its chromatic number is bounded.
\end{proof}

The fact that $K_t$-interval-minor-free ordered graphs have bounded chromatic number 
directly implies that the class of ordered graphs which do not contain some (arbitrary but fixed) ordered matching as a subgraph has bounded chromatic number (the two statements are in fact easily equivalent). This problem was introduced in \cite{ARU16} by Axenovich, Rollin and Ueckerdt, where they ask the question for some specific matchings. A far-reaching generalization was announced by Bria\'nski, Davies and Walczak~\cite{BDW}: for any ordered matching $M$, the class of ordered graphs which do not contain $M$ as an induced ordered subgraph is $\chi$-bounded. A study of the list chromatic number of ordered graphs avoiding an induced subgraph was also conducted by Hajebi, Li and Spirkl~\cite{HLS24}.

\subsection*{Hardness of detecting complete interval minors in ordered graphs}

In our main result, we give a polynomial approximate algorithm to detect whether an ordered graph contains a $K_t$ interval minor. For the hardness part, we show that deciding whether the complete interval minor number is at least $t$ is $\NP$-complete in general. Formally, we consider the following problem.

\begin{tcolorbox}[colback=white,colframe=black]
\textsc{Complete Interval Minor} \\
\textbf{Input:} An ordered graph $G$ and an integer $t$. \\
\textbf{Question:} Is $K_t$ an interval minor of $G$?
\end{tcolorbox}

\begin{theorem} \label{thm:CIM-NPC}
\textsc{Complete Interval Minor} is $\NP$-complete.
\end{theorem}

\begin{proof}
We first show that the problem is in $\NP$. To do so, observe that we can describe a $K_t$ interval minor model of $G$ by giving the intervals. Then, it can easily be checked in polynomial time that these intervals indeed form a model of $G$.

We show that the problem is $\NP$-hard by reduction from \textsc{Clique}. Consider an instance $(G, k)$ of \textsc{Clique}: we are asked whether $G$ contains a clique of size at least $k$.
We build an instance $(\hat{G}, t)$ of \textsc{Complete Interval Minor} as follows.
Write $V(G) = \{v_1, \ldots, v_n\}$. Consider the ordered graph $(\hat{G}, <)$ defined as \begin{itemize}
    \item $V(\hat{G}) = \{v_1, \ldots, v_n\} \cup \{u_1, \ldots, u_{n-1}\}$, where the $u_i$ are fresh vertices,
    \item $v_iv_j \in E(\hat{G})$ if and only if $v_iv_j \in E(G)$, $v_iu_j \in E(\hat{G})$ for every $i \in [n]$ and $j \in [n-1]$, and $u_iu_j \in E(\hat{G})$ for every $i \neq j \in [n-1]$,
    \item $v_1 < u_1 < v_2 < u_2 < \ldots < v_{n-1} < u_{n-1} < v_n$.
\end{itemize}
Informally, the ordered graph $(\hat{G}, <)$ is obtained from $G$ by arbitrarily ordering the vertices of $G$, and then adding a universal vertex between any two vertices of $G$. See \cref{fig:Ghat} for an illutration.
Finally, we set $t = n-1+k$.
Note that $((\hat{G}, <), t)$ can be built from $(G, k)$ in polynomial time.

We now prove that $G$ has a clique of size $k$ if and only if $K_t$ is an interval minor of $(\hat{G}, <)$.
Suppose first that $G$ has a clique of size $k$ induced by the vertices $v_{i_1}, \ldots, v_{i_k}$.
Then, the vertices $v_{i_1}, \ldots, v_{i_k}, u_1, \ldots, u_{n-1}$ induce a clique of size $k+n-1 = t$ in $\hat{G}$, hence $(\hat{G}, <)$ contains $K_t$ as an interval minor.
Conversely, suppose that $K_t$ is an interval minor of $(\hat{G}, <)$. Since there are only $n-1$ vertices $u_j$, there are at least $t - (n-1) = k$ intervals in the interval model of $K_t$ which do not contain a vertex $u_j$.
By definition of the order $<$, any interval that does not contain a vertex $u_j$ is of the form $\{v_i\}$.
Let $\{v_{i_1}\}, \ldots, \{v_{i_k}\}$ be these $k$ intervals. Then, the vertices $v_{i_1}, \ldots, v_{i_k}$ induce a clique of size $k$ in $G$.
\end{proof}

\begin{figure}[ht]
    \centering
    \begin{tikzpicture}[every node/.style={circle,draw,inner sep=1.5pt},
	arc/.style={red, bend left=30}
	]	
	
	\node (p1) at (-4,0) {};
	\node (p2) at (-3,0) {};
	\node (p3) at (-2,0) {};

	\draw[arc, black] (p1) to (p2);
	\draw[arc, black] (p2) to (p3);
	\node[draw = none] at (-3,-0.5){$G$};
	
	\node (1) at (0,0) {};
	\node[draw = none,fill = red]   (2) at (1,0) {};
	\node (3) at (2,0) {};
	\node[draw = none,fill = red]   (4) at (3,0) {};
	\node (5) at (4,0) {};
	
	\draw[arc, black] (1) to (3);
	\draw[arc, black] (3) to (5);
	\foreach \i in {2,4} 
		\foreach \j in {1,...,5} {
			\ifnum\i<\j
			\draw[arc] (\j.south) to (\i.south);
			\else
			\draw[arc] (\i.south) to (\j.south);
			\fi
		}
	\node[draw = none] at (2,-0.8){$\hat{G}$};
    \end{tikzpicture}
    \caption{A graph $G$ drawn with an arbitrary order and the corresponding ordered graph $\hat{G}$. The fresh vertices $u_1$ and $u_2$ are depicted in red, as well as the edges incident to them.}
    \label{fig:Ghat}
\end{figure}
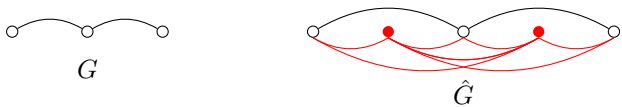

\subsection*{Perspectives}

The obvious next step is of course to obtain a more digest approximation factor than the triply exponential function $f$. We did not try very hard to show lower bounds, but failed to rule out constant factor approximation. A very natural problem to ask is the existence of an $\FPT$ algorithm to find a $K_t$ interval minor. This looks really challenging, as the only strategy seems to obtain a much more precise description of $K_t$-interval-minor-free ordered graphs which would allow dynamic programming. One of the main conclusions of this study is the versatility of delayed decompositions. It really suffices to apply them in a canonical way, and wonder whether the quotient graphs are simpler than the original one. For which (non necessarily ordered) graph classes does this machinery provide a bounded rank decomposition?

\subsection*{Organization of the paper}

In \cref{sec:delayed_decomp}, we formally define delayed decompositions, show how to compute them efficiently, and review some of their properties.
Then, in \cref{sec:delayed_rank}, we introduce a variant of rank tailored for algorithmic use, the \emph{delayed rank}, which is based on delayed decompositions. We study some of its basic properties, before proving that ordered graphs with large delayed rank contain large complete interval minors. 
Finally, in \cref{sec:algo}, we present the algorithm of \cref{thm:the-algo}. To bound its approximation factor, we prove a Ramsey-type result for interval minors. We also give a linear-time algorithm for testing whether an ordered graph contains $K_3$ as an interval minor.

%% file: delayed_decomp.tex
\section{Delayed decomposition} \label{sec:delayed_decomp}

This section focuses on the notion of delayed decomposition, a key tool for our approximation algorithm.
We start by introducing delayed decompositions in \cref{subsec:def-delayed}.
Then, in \cref{subsec:compute-delayed}, we prove that all (ordered) graphs admit so-called distinguishing delayed decompositions, and that they can be computed in linear time.
Finally, in \cref{subsec:trees}, we prove a variant of K\H{o}nig's lemma on trees, and explore its consequences related to delayed decompositions.

\subsection{Definition} \label{subsec:def-delayed}

We consider rooted trees $T$, and call the vertices of $T$ \emph{nodes}.
The \emph{ancestors} of a node $x$ are the nodes in the unique path from $x$ to the root $r$ of $T$,
and the \emph{parent} of $x$ is the first node on this path (other than $x$). The root has no parent.
We also speak of \emph{grandparents}, \emph{descendants}, \emph{children},
\emph{grandchildren}, \emph{siblings} (nodes with same parent),
and \emph{cousins} (non-sibling nodes with the same grandparent).
The set of leaves of a tree~$T$ is denoted by~$L(T)$.
For any node~$x \in V(T)$, we denote by~$L(x) \subseteq L(T)$ the set of leaves which are descendants of $x$.

An \emph{ordered tree}~$(T,<)$ is a rooted tree~$T$ equipped with a linear order~$<$ on $L(T)$,
such that for each node~$x \in V(T)$, the leaves~$L(x)$ form an interval of~$<$.
It is natural to think of~$<$ as a \emph{left-to-right} order on the leaves of $T$.
If~$(T,<)$ is an ordered tree and $x,y \in V(T)$ are not in an ancestor--descendant relationship,
then $L(x)$ and $L(y)$ are disjoint, and each of them is an interval for~$<$,
hence either $L(x) < L(y)$, or $L(y) < L(x)$.
We then naturally extend~$<$ to~$x,y$ by~$x < y$ in the former case, and~$y < x$ in the latter.
In particular, $<$ induces a linear order~$<_x$ on the children of any internal node~$x$.
Therefore, given a child $y$ of a node $x$, we can speak of the \emph{predecessor} and the \emph{successor} of $y$, and of \emph{consecutive} children of $x$.
We can also consider the \emph{first child} of $x$, which is the $<_x$-minimum child of $x$, and the \emph{last child} of $x$, defined analogously.

A \emph{delayed structured tree}~$\left(T,<,\{G_x\}_{x \in V(T)}\right)$
is an ordered tree~$(T,<)$, equipped with, for each node~$x \in V(T)$,
a graph~$G_x$ on the \emph{grandchildren} of~$x$. We will often refer to these graphs $G_x$ as the \emph{quotient graphs}.
This is analogous to the trees describing substitutions (module decomposition trees),
except that the graphs~$G_x$ are defined on the grandchildren instead of the children, hence ``delayed''.
We add the technical requirement that each leaf is a~single child (with no siblings),
so that whenever~$x \neq y$ are leaves, their closest ancestor is at distance at least~2.

The \emph{realization} of a delayed structured tree~$\left(T,<,\{G_x\}_{x \in V(T)}\right)$ is the ordered graph ${G_T=((L(T),E_T),<)}$,
where for two leaves~$x,y$ of $T$, we have the edge $xy\in E_T$ if and only if $x'y'\in E(G_z)$,
where~$z$ is the closest ancestor of~$x,y$, and~$x',y'$ are the grandchildren of~$z$
which are ancestors of~$x,y$ respectively.
If $(G, <)$ is the realization~$\left(T,<,\{G_x\}_{x \in V(T)}\right)$, we say that~$\left(T,<,\{G_x\}_{x \in V(T)}\right)$ is a \emph{delayed decomposition} of $(G, <)$. See \cref{fig:delayed} for an example.
\begin{figure}[ht]
	\centering
	\begin{tikzpicture}[scale=0.4,auto=left,every node/.style={circle,draw,fill=black, inner sep=1.5pt }]

		\node (1) at (-4.5,11) {};
		\node (2) at (-8.5,9) {};
		\node (3) at (0,9) {};

        \node [fill=white] (4) at (-8.5,6) {};
		\node (5) at (-4.5,6) {};
		\node (6) at (0,6) {};
		\node (7) at (5,6) {};

		\node [fill=white] (8) at (-4.5,3) {};
		\node (9) at (-1.5,3) {};
		\node (10) at (0,3) {};
		\node (11) at (1.5,3) {};
		\node (12) at (4.25,3) {};
		\node (13) at (5.75,3) {};

		\node [fill=white] (14) at (-1.5,1) {};
		\node (15) at (-0.5,1) {};
		\node (16) at (0.5,1) {};
		\node [fill=white] (17) at (1.5,1) {};
		\node (18) at (3.75,1) {};
		\node (19) at (4.75,1) {};
		\node [fill=white] (20) at (5.75,1) {};

		\node [fill=white] (21) at (-0.5,-1) {};
		\node [fill=white] (22) at (0.5,-1) {};
		\node [fill=white] (23) at (3.75,-1) {};
		\node [fill=white] (24) at (4.75,-1) {};

		\draw(1) -- (2);
		\draw(1) -- (3);
		\draw(2) -- (4);
		\draw(3) -- (5);
		\draw(3) -- (6);
		\draw(3) -- (7);
		
		\draw(5) -- (8);
		\draw(6) -- (9);
		\draw(6) -- (10);
		\draw(6) -- (11);
		\draw(7) -- (12);
		\draw(7) -- (13);
		\draw(9) -- (14);
		\draw(10) -- (15);
		\draw(10) -- (16);
		\draw(11) -- (17);
		\draw(12) -- (18);
		\draw(12) -- (19);	
		\draw(13) -- (20);
		
		\draw(15) -- (21);
		\draw(16) -- (22);
		\draw(18) -- (23);
		\draw(19) -- (24);

		\draw[blue] (4) to [bend left=30] (5);
		\draw[blue] (4) to [bend left=30] (7);
		\draw[orange] (8) to[bend left=50] (12);
		\draw[orange] (9) to[bend left=50] (13);
		\draw[orange] (11) to[bend left=50] (12);
		\draw[orange] (11) to[bend left=50] (13);
		\draw[cyan] (14) to[bend left=50] (15);
		\draw[cyan] (16) to[bend left=50] (17);
		\draw[magenta] (21) to[bend left=50] (22);

		\foreach \i in {1,...,9} 
		\node[fill = white] (A\i) at (\i*2 + 8, 4) {};

		\foreach \i in {2,7,8,9}
		\draw[blue] (A1) to [bend left=50] (A\i);

		\foreach \i in {7,8}
		\draw[orange] (A2) to[bend right=50] (A\i);
		\draw[orange] (A3) to[bend right=50] (A9);
		\foreach \i in {7,8,9}
		\draw[orange] (A6) to[bend right=50] (A\i);

		\draw[cyan] (A3) to[bend left=50] (A4);
		\draw[cyan] (A5) to[bend left=50] (A6);

		\draw[magenta] (A4) to[bend left=50] (A5);        
		
	\end{tikzpicture}
	\caption{A delayed structured tree and its realization. The edges in the realization are colored according to their corresponding quotient graph.}
	\label{fig:delayed}
\end{figure}

A delayed decomposition $\left(T, <, \{G_x\}_{x \in V(T)}\right)$ of an ordered graph $(G, <)$ is \emph{distinguishing} if it satisfies the following property:
For every node $x$ which has at least $3$ children, for every two consecutive children $x_1, x_2$ of $x$, there exists some vertex $v \in V(G) \setminus L(x)$ which is adjacent to all vertices in one of $L(x_1), L(x_2)$ and to no vertex in the other.

As we shall see in \cref{thm:compute-delayed}, every graph has a distinguishing delayed decomposition, which can be computed in linear time. When we talk about \emph{the} distinguishing delayed decomposition of a graph, we mean the one computed by \cref{thm:compute-delayed}.

\begin{lemma}[{\cite[Lemma~2.1]{BT25}}] \label{lem:delayed=subst+edge}
Let $(G, <)$ be the realization of a delayed structured tree~$\left(T,<,\{G_x\}_{x \in V(T)}\right)$. 
Then, $(G, <)$ is the edge union of two ordered graphs, each of which can be obtained by substitutions from the quotient graphs $\{G_x\}_{x \in V(T)}$.
\end{lemma}

\subsection{Computing a distinguishing delayed decomposition} \label{subsec:compute-delayed}

In this section, we prove that every ordered graph has a distinguishing delayed decomposition, which can be computed in linear time.

To discuss the running time of the algorithm, we first detail how we store ordered graphs. Throughout this article, we assume that the edge set is stored as a list of edges, and the vertex set as a list of vertices. There are two natural ways to store the order on the vertex set. The first one is to store the \emph{ordered} set of vertices explicitly, for instance by assuming that the list of vertices is sorted according to the order. In particular, after a $O(n)$-time precomputation, all vertices can store their rank in the order. We call this representation \emph{explicit}. The second one is to assume that the order can be determined from the labels of the vertices, which is the case for instance if the vertices are ordered by increasing labels. We call this representation \emph{implicit}.
Observe that an explicit representation can be computed in time $O(n\log(n))$ from an implicit representation by sorting the vertex set according to the order, and then storing the sorted list explicitly.
Given an explicit representation of an ordered graph, we can relabel all the vertices in time $O(n)$ so that the vertex set is $[n]$, equipped with its natural linear order.

\begin{theorem}\label{thm:compute-delayed}
There is an algorithm which, given as input an explicit ordered graph~${(G,<)}$ with $n$ vertices and $m$ edges, computes in time $O(m+n)$ a distinguishing delayed decomposition $\left(T, <, \{G_x\}_{x \in V(T)}\right)$ of $(G, <)$.
\end{theorem}

This algorithm was first described in \cite{BT25}, where delayed decompositions were introduced, and it was shown in \cite{BBGT24} how to implement this algorithm in linear time in the context of permutations. 
The proof of \cref{thm:compute-delayed} is just a translation of the implementation of \cite{BBGT24} in the context of ordered graphs. We first need some standard algorithmic results.

If $A$ is an array, an \emph{interval} of $A$ is a set of consecutive entries of $A$.

\begin{lemma}[\cite{bender2000LCA,harel1984LCA,Schieber1988LCA}] \label{lem:rmq}
There is an algorithm which, given as input an array $A$ of size $n$, after a $O(n)$-time preprocessing, can return the minimum and maximum (and their positions) of any interval of $A$ in constant time.
\end{lemma}

If $T$ is a rooted tree and $u, v$ are nodes of $T$, the \emph{last common ancestor} (LCA) of $u$ and $v$ is the only node on the path between $u$ and $v$ in $T$ which is an ancestor of both $u$ and $v$.
Given two nodes $u$ and $v$ of $T$, the answer to the \emph{extended LCA query} for $u$ and $v$ is the data of the LCA $z$ of $u$ and $v$, together with the children and grandchildren of $z$ which are the ancestors of $u$ and of $v$ (if they exist).

\begin{lemma}[\cite{bender2000LCA}]\label{lem:LCA}
There is an algorithm which, given as input a rooted tree $T$ on $n$ vertices, after a $O(n)$-time preprocessing, can answer any extended LCA query on $T$ in constant time.
\end{lemma}

We now have all the tools to prove \cref{thm:compute-delayed}.

\begin{proof}[\textit{Proof of \cref{thm:compute-delayed}}]
We first describe the algorithm and prove its correctness. We then show how to implement it to run in time $O(m+n)$.

When we talk about intervals, we always mean intervals for the order $<$.
We start by building the tree $T$, whose leaves will be the vertices of $G$. We construct $T$ inductively starting from the root by specifying for every node $x$ of $T$ the interval $L(x) \subseteq V(G)$ corresponding to its eventual set of leaf descendants.
Throughout the construction, we ensure the following \emph{consistency property}: For every node $x \in V(T)$, if $v \in V(G) \setminus L(x)$ and $y$ is a descendant of $x$ then either $v$ is adjacent to all vertices in $L(y)$ or $v$ is adjacent to no vertex in $L(y)$.

We start with the root $r$ of $T$ and set $L(r) = V(G)$.
Suppose that we are considering a node $x$ for which the set $L(x)$ is already defined.
We distinguish several cases. \begin{itemize}
    \item If $|L(x)| = 1$, we add one child $y$ to $x$, set $L(y) = L(x)$ and stop the construction here for this branch (in particular, we will not consider the node $y$).
    Observe that the consistency property trivially holds for $x$ in that case.
    \item If $|L(x)| \geq 2$ and $L(x)$ is a module in $G$ (meaning that every $v \in V(G) \setminus L(x)$ is either adjacent to all vertices in $L(x)$ or to no vertex in $L(x)$), we add two children $y_1$ and $y_2$ to $x$ and set $L(y_1) = \{u\}$, where $u$ is the $<$-smallest element of $L(x)$, and $L(y_2) = L(x) \setminus \{u\}$.
    Observe that if $v \in V(G) \setminus L(x)$ then $v$ is either adjacent to all vertices in $L(x)$ or to no vertex in $L(x)$, so the same holds for $L(y_1)$ and $L(y_2)$ and the consistency property holds for $x$.
    Note also that this is what happens at the root $r$ of $T$.
    \item Otherwise, for $u_1, u_2 \in L(x)$, write $u_1 \sim u_2$ if every $w$ in the interval between $u_1$ and $u_2$ satisfies that $u_1, u_2$ and $w$ have the same neighborhood in ${V(G)\setminus L(x)}$.
    Observe that the relation $\sim$ is an equivalence relation on $L(x)$.
    The equivalence classes for $\sim$ form intervals of $L(x)$, called \emph{local modules}.
    Let $I_1, \ldots, I_k$ be the local modules of $L(x)$. Note that $k \geq 2$ otherwise $L(x)$ would be a module in $G$. We add $k$ children $y_1, \ldots, y_k$ to $x$, and set $L(y_i) = I_i$ for every $i \in [k]$.
    By definition of the local modules, for every $v \in V(G) \setminus L(x)$ and every child $y_i$ of $x$, either $v$ is adjacent to all vertices in $L(y_i)$ or $v$ is adjacent to no vertex in $L(y_i)$, so the consistency property holds for $x$.
    Finally, we prove the distinguishing property. Consider two consecutive children $y_i, y_{i+1}$ of $x$, and let $u_i$ be the $<$-largest vertex of $I_i$ and $u_{i+1}$ be the $<$-smallest vertex of $I_{i+1}$. Since $u_i \not \sim u_{i+1}$ and since $u_i$ and $u_{i+1}$ are consecutive, there exists some vertex $v \in V(G) \setminus L(x)$ such that $v$ is adjacent to exactly one of $u_i,u_{i+1}$.
    Since $I_i$ and $I_{i+1}$ are local modules in $L(x)$ and since $v \notin L(x)$ then $v$ is adjacent to all vertices in one of $I_i, I_{i+1}$ and to no vertex in the other, as desired.
\end{itemize}

Note that as long as $|L(x)| \geq 2$, the node $x$ has at least $2$ children, each with a strictly smaller set of leaf descendants. This proves that the construction of $T$ terminates.
It follows from the construction that for every vertex $u \in V(G)$, the set of nodes $\{x \in V(T) : u \in L(x)\}$ induces a root-to-leaf path in $T$.
Furthermore, if $x$ is a leaf of $T$ then $L(x)$ is a singleton. 
Therefore, the leaves of $T$ naturally correspond to the vertices of $G$. From now on, we consider that the set of leaves of $T$ is the set of vertices of $G$.
The order $<$ on $V(G)$ then naturally gives an order on $L(T)$.

We now define the quotient graphs $\{G_x\}_{x \in V(T)}$. The vertex set of the graph $G_x$ is the set of grandchildren of $x$ in $T$. If $y_1, y_2$ are grandchildren of $x$ which are not siblings, we add an edge $y_1y_2$ to $G_x$ if and only if there is an edge between $L(y_1)$ and $L(y_2)$. 

We now argue that $(G, <)$ is equal to the realization $G_T$ of $\left(T, <, \{G_x\}_{x \in V(T)}\right)$.
First, ${V(G_T) = L(T) = V(G)}$, and they are equipped with the same linear order.
Consider now an edge $uv \in E(G)$. The vertices $u$ and $v$ are leaves of $T$. Let $z$ be their last common ancestor in $T$, and let $y_u$ and $y_v$ be the grandchildren of $z$ which are the respective ancestors of $u$ and $v$. Then, $y_u$ and $y_v$ are cousins and the edge $uv$ is an edge between $L(y_u)$ and $L(y_v)$ so by definition of $G_z$, there is an edge $y_uy_v$ in $G_z$. Then, by definition of $G_T$ we have $uv \in E(G_T)$.
Conversely, consider an edge $uv \in E(G_T)$. The vertices $u$ and $v$ are leaves of $T$. Let $z, y_u, y_v$ be defined as previously. Then, $y_u$ and $y_v$ are cousins and since $uv \in E(G_T)$ then $y_uy_v \in E(G_z)$. Thus, by definition of $G_z$, there is an edge between some $u' \in L(y_u)$ and some $v' \in L(y_v)$.
By the consistency property for the parent $x_u$ of $y_u$, we get that $v'$ is adjacent in $G$ to every vertex in $L(y_u)$, and in particular to $u$. By the consistency property for the parent $x_v$ of $y_v$, $u$ is adjacent in $G$ to every vertex in $L(y_v)$, and in particular to $v$, so $uv \in E(G)$.

We now explain how to implement this algorithm to run in time $O(m+n)$.
Since the ordered graph is stored explicitly, we can relabel the vertices in time $O(n)$ so that $V(G) = [n]$.
Using bucket sort, we can compute the sorted adjacency lists of all vertices in time $O(m + n)$.
For every $i \in [n-1]$, let $m(i)$ (resp. $M(i)$) be the minimum (resp. maximum) vertex which is adjacent to one of $i, i+1$ and not adjacent to the other, or $\infty$ (resp. $0$) if there is no such vertex.
From the sorted adjacency lists, we can compute all $m(i)$ and $M(i)$ in time $O(m+n)$.
We then preprocess the arrays $m$ and $M$ as in \cref{lem:rmq} to be able to answer maximum and minimum queries on them in constant time.

We construct the tree $T$ as described above.
We start by creating the root node $r$, and set $L(r) = [n]$.
Suppose that we are considering some node $x$ with $L(x)$ already defined, say $L(x) = [a, b]$.
If $|L(x)| = 1$, we proceed as indicated.
Otherwise, we compute the local modules $I_1, \ldots, I_{k+1}$ of $L(x)$ in time $O(k)$.
To do so, let $\{j_1 < j_2 < \ldots < j_k\}$ be the set of all $i \in L(x) \setminus \{b\}$ such that either $m(i) < a$ or $M(i) > b$.
Using \cref{lem:rmq} on the arrays $m$ and $M$, we can find all $j_l$ in time $O(k)$.
If $k=0$ then $L(x)$ is a module and we proceed as described.
If $k \geq 1$ then $L(x)$ is not a module, and the local modules of $L(x)$ are $[a, j_1], [j_1+1, j_2], \ldots, [j_k+1, b]$.
In that case, we again proceed as explained above.

Since the leaves of $T$ are the vertices of $G$ then $T$ has $n$ leaves.
Furthermore, every internal node of $T$ has at least $2$ children, except for the parents of the leaves which all have $1$ child. 
Then, the total number of nodes of $T$ is at most $3n-1$.
Furthermore, if a node $x$ has $k$ children, the time spent while considering $x$ is $O(k)$.
Thus, the time spent for the creation of $T$ is $O(n)$.

We now build the graphs $\{G_x\}_{x \in V(T)}$.
For every $x \in V(T)$, initialize $G_x$ as being the edgeless graph whose vertex set is the set of grandchildren of $x$. Importantly, we can store the vertex set of $G_x$ sorted according to $<$, so that all graphs $G_x$ will be stored explicitly.
For every vertex $u \in V(G)$, we keep a pointer to the leaf of $T$ corresponding to $u$.
We preprocess $T$ as in \cref{lem:LCA} to be able to answer extended LCA queries in constant time.
Then, we iterate over all edges $uv$ of $G$. For each such edge, we find the last common ancestor $z$ of $u$ and $v$ in $T$, and the grandchildren $y_u$ and $y_v$ of $z$ which are the respective ancestors of $u$ and $v$ in $T$. This can be done in constant time using \cref{lem:LCA}. Then, we add an edge between $y_u$ and $y_v$ in $G_z$.
Overall, this takes time $O(m + n)$.
\end{proof}

\begin{remark}\label{rmk:nb-quotients}
If $(G, <)$ is an ordered graph with $n$ vertices and $m$ edges, the distinguishing delayed decomposition computed by \cref{thm:compute-delayed} has the property that the total number of edges over all quotient graphs is at most $m$. Furthermore, the number of quotient graphs is equal to the number of nodes of $T$ which have grandchildren, which is at most $n-1$. Observe also that, by construction, for every node $x$ of $T$ and every child $y$ of $x$, the children of $y$ form an independent set in $G_x$.
\end{remark}

\subsection{A result on trees and its consequences} \label{subsec:trees}

A classical result on trees states that every tree with a very large number of leaves contains either a node with large degree, or a path containing many nodes of degree at least 3. The goal of this section is to prove a generalization of this statement in the context of ordered trees. More precisely, consider an ordered tree whose leaves are grouped into disjoint intervals. We prove that if there is a very large number of intervals, either there is a node which ``splits'' a large number of intervals between its children, or there is a long path which progressively ``peels'' a large number of intervals.

Given an ordered tree $(T, <)$, an \emph{interval family} of $(T,<)$ is a set $\mathcal I$ of disjoint intervals of the linear order $(L(T),<)$. A node $x\in V(T)$ is \emph{$b$-branching} if there is a subset $\mathcal I' \subseteq \mathcal{I}$ of $b$ intervals such that every interval of $\mathcal I '$ is included in $L(x)$ and there is no child $y$ of $x$ such that $L(y)$ intersects two intervals of $\mathcal I '$. See \cref{fig:branching/interval} for an illustration. An \emph{$\ell$-interval path} is a sequence of intervals $I_1,\dots ,I_\ell$ such that there exists nodes $x_1,\dots ,x_\ell$ of $T$ such that:
\begin{itemize}
    \item $L(x_j)$ contains $I_j\cup \dots \cup I_\ell$ for all $j=1,\dots ,\ell$.
    \item $L(x_j)\cap I_{j-1}=\emptyset$ for all $j=2,\dots ,\ell$.
\end{itemize}

Note that the nodes $x_j$ are pairwise distinct, and $x_j$ is a descendant of $x_i$ whenever $i<j$. Also, when $j\geq3$, the parent $p(x_j)$ of $x_j$ satisfies $L(p(x_j))\cap I_{j-2}=\emptyset$ since $p(x_j)$ is a (not necessarily proper) descendant of $x_{j-1}$.

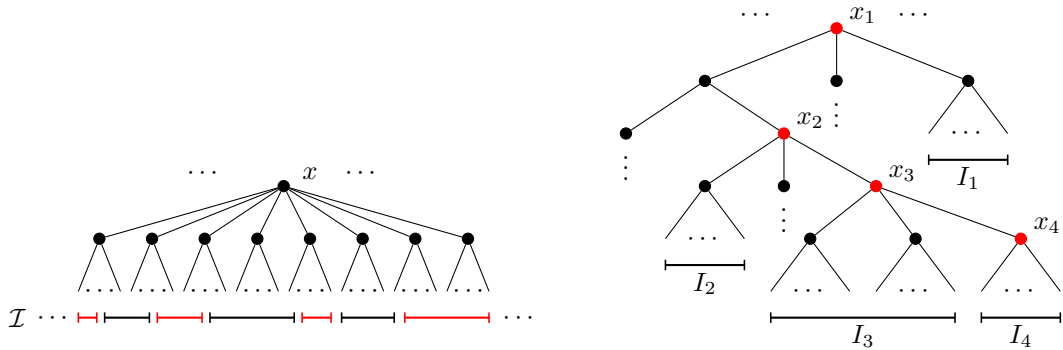
\begin{figure}[ht]
	\centering
	\begin{tikzpicture}[scale=0.35,auto=left,every node/.style={circle,draw,fill=black, inner sep=1.5pt }]

		\node (0) at (9,0) {};
		\draw (10,0.5) node[draw=none,fill=none] {$x$};
		\draw (12,0.5) node[draw=none,fill=none] {$\cdots$};
		\draw (6,0.5) node[draw=none,fill=none] {$\cdots$};

		\foreach \i in {1,...,8} 
		\node (\i) at (\i*2 , -2) {};

		\foreach \i in {1,...,8} {
		\draw (0) -- (\i);
		\draw (\i) -- (\i*2-0.8 , -4);
		\draw (\i) -- (\i*2+0.8 , -4);
		\draw (\i*2+0.1,-4) node[draw=none,fill=none] {$\cdots$};
		}

		\draw (-1.1,-5) node[draw=none,fill=none] {$\mathcal{I}$};
		\draw (0.3,-5) node[draw=none,fill=none] {$\cdots$};
		\draw (18,-5) node[draw=none,fill=none] {$\cdots$};
		\foreach \i in {1} {
		    \draw[red,line width=0.8pt] (\i*2-0.8,-5) --(\i*2-0.1,-5);
		    \draw[red,line width =0.6pt] (\i*2-0.8,-5.2) --(\i*2-0.8,-4.8);
		    \draw[red,line width = 0.6pt]
		    (\i*2-0.1,-5.2) -- (\i*2-0.1,-4.8);
		}
		\foreach \i in {3} {
			\draw[red,line width=0.8pt] (\i*2-1.8,-5) -- (\i*2-0.1,-5);
			\draw[red,line width=0.6pt] (\i*2-1.8,-5.2) -- (\i*2-1.8,-4.8);
			\draw[red,line width=0.6pt] ((\i*2-0.1,-5.2) -- (\i*2-0.1,-4.8);
		}
		\foreach \i in {2} {
		\draw[line width=0.8pt] (\i*2-1.8,-5) -- (\i*2-0.1,-5);
		\draw[line width=0.6pt] (\i*2-1.8,-5.2) -- (\i*2-1.8,-4.8);
		\draw[line width=0.6pt] ((\i*2-0.1,-5.2) -- (\i*2-0.1,-4.8);
		}
		\draw[line width=0.8pt] (5*2-3.8,-5) -- (5*2-0.6,-5);
		\draw[line width=0.6pt] (5*2-3.8,-5.2) -- (5*2-3.8,-4.8);
		\draw[line width=0.6pt] (5*2-0.6,-5.2) -- (5*2-0.6,-4.8);
		
		\draw[red,line width=0.8pt] (6*2-2.3,-5) -- (6*2-1.2,-5);
		\draw[red,line width=0.6pt] (6*2-2.3,-5.2) -- (6*2-2.3,-4.8);
		\draw[red,line width=0.6pt] (6*2-1.2,-5.2) -- (6*2-1.2,-4.8);
		
		\draw[line width=0.8pt] (6*2-0.8,-5) -- (6*2+1.2,-5);
		\draw[line width=0.6pt] (6*2-0.8,-5.2) -- (6*2-0.8,-4.8);
		\draw[line width=0.6pt] (6*2+1.2,-5.2) -- (6*2+1.2,-4.8);
		
		\draw[red,line width=0.8pt] (7*2-0.4,-5) -- (7*2+2.8,-5);
		\draw[red,line width=0.6pt] (7*2-0.4,-5.2) -- (7*2-0.4,-4.8);
		\draw[red,line width=0.6pt] (7*2+2.8,-5.2) -- (7*2+2.8,-4.8);

		\node[red] (T0) at (30,6) {};
		\draw (31,6.5) node[draw=none,fill=none] {$x_1$};
		\draw (33,6.5) node[draw=none, fill=none]{$\cdots$};
		\draw (27,6.5) node[draw=none,fill=none]{$\cdots$};
		\node (T1) at (25,4) {};
		\node (T2) at (35,4) {};
		\node (T3) at (30,4) {};

		\node (T4) at (22,2) {};
		\node[red] (T5) at (28,2) {};
		\draw (29,2.5) node[draw=none,fill=none] {$x_2$};

		\node[red] (T6) at (31.5,0) {};
		\draw (32.5,0.5) node[draw=none,fill=none] {$x_3$};
		\node (T7) at (25,0) {};
		\node (T8) at (28,0) {};

		\node (T9) at (33,-2) {};
		\node (T10) at (29,-2) {};
		\node[red] (T11) at (37,-2) {};
		\draw (38,-1.5) node[draw=none,fill=none] {$x_4$};

		\draw (T0) -- (T1);
		\draw (T0) -- (T2);
		\draw (T0) -- (T3);
		\draw (T1) -- (T4);
		\draw (T1) -- (T5);
		\draw (T5) -- (T6);
		\draw (T5) -- (T7);
		\draw (T5) -- (T8);
		\draw (T6) -- (T9);
		\draw (T6) -- (T10);
		\draw (T6) -- (T11);
		
		\draw (T9) -- (34.5,-4);
		\draw (T9) -- (31.5,-4);
		\draw (T10) -- (27.5,-4);
		\draw (T10) -- (30.5,-4) ;
		\draw (T11) -- (35.5,-4);
		\draw (T11) -- (38.5,-4) ;
		\draw (T2) -- (36.5,2);
		\draw (T2) -- (33.5,2) ;
		\draw (T7) -- (26.5,-2);
		\draw (T7) -- (23.5,-2) ;

		\draw (30,3) node[draw=none,fill=none] {$\vdots$};
		\draw (22,1) node[draw=none,fill=none] {$\vdots$};
		\draw (28,-1) node[draw=none,fill=none] {$\vdots$};

		\draw (37,-4) node[draw=none,fill=none] {$\cdots$};
		\draw (33,-4) node[draw=none,fill=none] {$\cdots$};
		\draw (29,-4) node[draw=none,fill=none] {$\cdots$};
		\draw (35,2) node[draw=none,fill=none] {$\cdots$};
		\draw (25,-2) node[draw=none,fill=none] {$\cdots$};

		\draw[line width=0.8pt] (35.5,-5) -- (38.5,-5);
		\draw[line width=0.6pt] (35.5,-5.2) -- (35.5,-4.8);
		\draw[line width=0.6pt] (38.5,-5.2) -- (38.5,-4.8);
		\draw (37,-5.7) node[draw=none,fill=none] {$I_4$};
		
		\draw[line width=0.8pt] (27.5,-5) -- (34.5,-5);
		\draw[line width=0.6pt] (27.5,-5.2) -- (27.5,-4.8);
		\draw[line width=0.6pt] (34.5,-5.2) -- (34.5,-4.8);
		\draw (31,-5.7) node[draw=none,fill=none] {$I_3$};
		
		\draw[line width=0.8pt] (33.5,1) -- (36.5,1);
		\draw[line width=0.6pt] (33.5,1.2) -- (33.5,0.8);
		\draw[line width=0.6pt] (36.5,1.2) -- (36.5,0.8);
		\draw (35,0.3) node[draw=none,fill=none] {$I_1$};
		
		\draw[line width=0.8pt] (23.5,-3) -- (26.5,-3);
		\draw[line width=0.6pt] (23.5,-3.2) -- (23.5,-2.8);
		\draw[line width=0.6pt] (26.5,-3.2) -- (26.5,-2.8);
		\draw (25,-3.7) node[draw=none,fill=none] {$I_2$};
	\end{tikzpicture}
	\caption{A 4-branching vertex $x$, with $\mathcal{I}'$ being the set of red intervals, and a 4-interval path $I_1,I_2,I_3,I_4$.}
	\label{fig:branching/interval}
\end{figure}

\begin{lemma}\label{lem:branchtree}
Given an interval family $\mathcal I$ of size $2(b+2)^{t}$ of an ordered tree $(T,<)$, there is either a $b$-branching node or a $t$-interval path in $T$.
\end{lemma}

 \begin{proof}
For every node $x$ of $T$, let $w(x)$ be the number of intervals of $\mathcal I$ which are entirely contained inside $L(x)$. Let $B(x)$ be the set of children $y$ of $x$ such that $w(y)\neq 0$.
Note that $x$ is (at least) $|B(x)|$-branching, and so we can conclude if $|B(x)|\geq b$. Therefore, we can assume that $|B(x)|<b$ for every $x \in V(T)$. 

Consider a node $x$ with $w(x)\geq 2b^2$.
Let $\mathcal I_1$ be the set of intervals which are contained inside some $L(y)$ for $y\in B(x)$ and $\mathcal I_2$ be the set of intervals which are contained inside $L(x)$ but are not in $\mathcal I_1$. 
Note that if $|{\mathcal I_2}|\geq 2b-1$, the node $x$ is $b$-branching. Indeed, in that case it suffices to order
$\mathcal I_2$ with respect to $<$, and to select every other interval to form a $b$-branching subfamily $\mathcal I'$.

So there is a child $y$ of $x$ such that (using that $w(x) \geq 2b^2$ in the last inequality): $$w(y)\geq |\mathcal{I}_1|/|B(x)| \geq (w(x)-2b+1)/(b-1)\geq w(x)/b.$$
Starting from the root $r$, we define a sequence of vertices $(r=x_1,\dots,x_{t})$ forming a descending chain in $T$, such that $w(x_{i+1})\leq w(x_{i})-3$ for every $i \in [t-1]$ and $w(x_i)\geq 2(b+2)^{t+1-i}$ for every $i \in [t]$.

Start by setting $x_1$ to be the root $r$, and note that $w(x_1) = 2(b+2)^{t}$.
Suppose that we already defined $x_1, \ldots, x_{i-1}$ satisfying the desired property, with $i \leq t$. Consider a descendant $x_{i}$ of $x_{i-1}$ with maximum $w(x_i)$, such that ${w(x_i)\leq w(x_{i-1})-3}$, and among them one which is as close to the root as possible in $T$.
By definition of $x_i$, we have: $$w(p(x_i)) \geq w(x_{i-1}) -2 \geq 2(b+2)^{t+2-i} - 2 \geq 2(b+2)^2 -2 \geq 2b^2.$$ 
Thus, $p(x_i)$ has a child $y$ with $$w(y)\geq w(p(x_i))/b \geq (w(x_{i-1}) - 2)/b.$$ 
Therefore, by maximality of $w(x_i)$, we have: $$w(x_i) \geq (w(x_{i-1}) -2)/b \geq (2(b+2)^{t+2-i}-2)/b \geq 2(b+2)^{t+1-i}.$$

For every $i \in [t-1]$, since $w(x_{i+1}) \leq w(x_{i}) - 3$, there are at least three intervals which are entirely contained in $L(x_i)$ and which are not entirely contained in $L(x_{i+1})$. Therefore, there is an interval $I_i$ which is entirely contained in $L(x_i)$ and such that $I_i \cap L(x_{i+1}) = \emptyset$.
Finally, $w(x_t) \geq 2(b+2) > 0$ so there exists an interval $I_t \subseteq L(x_t)$.
\end{proof}

If $(T, <, \{G_x\}_{x \in V(T)})$ is a delayed structured tree, a leaf $y$ of $T$ is \emph{$h$-heavy} if there are at least $h$ ancestors $x$ of $y$ which are not isolated in the graph $G_{p^2(x)}$, where $p^2(x)$ is the grandparent of $x$ in $T$.

\begin{lemma}\label{lem:cliquepath}
Let $(T, <, \{G_x\}_{x \in V(T)})$ be a delayed structured tree and $(G_T, <)$ be its realization. Let $\mathcal{I}$ be an interval family of $(T, <)$ forming a complete interval minor in $(G_T, <)$. If $\mathcal{I}$ has a $(2t-1)$-interval path in $T$, then there is a $(2t-3)$-heavy leaf in $(T, <, \{G_x\}_{x \in V(T)})$.
\end{lemma}

\begin{proof}
    Let $I_1, \ldots, I_{2t-1}$ be the intervals of the $(2t-1)$-interval path of $\mathcal{I}$ in $T$, and $x_1, \ldots, x_{2t-1}$ be nodes of $T$ such that for all $j \in [2t-1]$, $L(x_j)$ contains $I_j\cup \dots \cup I_{2t-1}$, and for all $j \in [\![2, 2t-1]\!]$, $L(x_j)\cap I_{j-1}=\emptyset$.
    
    Let $y \in L(T)$ be any leaf in $I_{2t-1}$.
    We show that $y$ is $(2t-3)$-heavy.
    Let $i \in [2t-3]$. Since $\mathcal{I}$ forms a complete interval minor in $(G_T, <)$, there is an edge between some vertex $y_i \in I_i$ and $I_{2t-1}$. Let $z_i$ be the deepest node of $T$ such that $I_{2t-1} \cup \{y_i\} \subseteq L(z_i)$.
    Then, $z_i$ is a descendant of $x_i$ and a proper ancestor of $x_{i+1}$ (hence all $z_i$ are distinct).
    Furthermore, since $z_i$ is a proper ancestor of $x_{i+1}$ then the grandchild $w_i$ of $z_i$ such that $y \in L(w_i)$ is an ancestor of $x_{i+2}$.
    In particular, $I_{2t-1} \subseteq L(x_{2t-1}) \subseteq L(x_{i+2}) \subseteq L(w_i)$.
    Let $w'_i$ be the grandchild of $z_i$ such that $y_i \in L(w'_i)$. 
    By maximality of the depth of $z_i$, $w_i$ and $w'_i$ are cousins in $T$.
    Since there is an edge between $y_i$ and $I_{2t-1}$ in $G_T$, there is an edge between $w_i$ and $w'_i$ in $G_{z_i}$.
    Thus, we found $2t-3$ distinct ancestors $w$ of $y$ which are not isolated in $G_{p^2(w)}$, i.e. $y$ is $(2t-3)$-heavy.
\end{proof}

\begin{lemma}\label{lem:cliquepath2}
Let $(T, <, \{G_x\}_{x \in V(T)})$ be a delayed structured tree containing a $(2t-3)$-heavy leaf.
Then, its realization $(G_T, <)$ contains a clique of size $t$ as a subgraph.
\end{lemma}

\begin{proof}
    Let $y$ be a $(2t-3)$-heavy leaf of $T$.
    Let $x_1, \ldots, x_{2t-3}$ be ancestors of $y$ such that each $x_i$ is not isolated in the graph $G_{p^2(x_i)}$. Up to renaming them, we can assume that $x_i$ is a proper ancestor of $x_j$ whenever $i < j$.
    We build a sequence $(y_0, \ldots, y_{t-1})$ of vertices of $G_T$ with the following properties:
    \begin{itemize}
        \item For every $i \in [\![0, t-2]\!]$, $y_i$ is adjacent to every vertex in $L(x_{2i+1})$, and
        \item For every $i \in [t-1]$, $y_i \in L(x_{2i-1})$.
    \end{itemize}
    Note that these properties immediately imply that $y_0, \ldots, y_{t-1}$ induce a clique of size $t$ in $G_T$.
    
    Let $i \in [\![0, t-2]\!]$ and let $x'_i$ be a neighbor of $x_{2i+1}$ in the graph $G_{p^2(x_{2i+1})}$.
    Let $y_{i}$ be an arbitrary vertex in $L(x'_{i})$. Then, $y_i$ is adjacent to all vertices in $L(x_{2i+1})$.
    If $i > 0$ then $p^2(x_{2i+1})$ is a descendant (not necessarily proper) of $x_{2i-1}$ so $y_i \in L(x_{2i-1})$.
    Finally, we choose $y_{t-1}$ to be any vertex of $L(x_{2t-3})$.
\end{proof}

%% file: delayed_rank.tex
\section{Delayed rank} \label{sec:delayed_rank}

The notion of rank is natural and convenient for structural analysis, but not so much for algorithmic purposes, since there does not seem to be a simple way of computing it, or even approximating it. 
To overcome this issue, we introduce an alternative to the rank, based on delayed decompositions, which we call the \emph{delayed rank}. This delayed rank can easily be computed, and shares some key structural properties with the rank. It is the key notion for our approximation algorithm.

\subsection{Definition and first properties}

We add some information to the delayed structured trees. Consider a delayed structured tree $(T,<,\{G_x\}_{x \in V(T)})$. We label the nodes $x$ of $T$ as follows (see \cref{fig:delayed_labeling}):

\begin{itemize}
    \item If $x$ does not have a grand-parent, we label it with $\emptyset$ (the first two layers of the tree).
    \item Else, if there is a cousin $x'$ of $x$ such that $x < x'$ and there is an edge $xx'$ in $G_{p^2(x)}$ then we label $x$ with $R$.
    \item Otherwise, if there is a cousin $x'$ of $x$ such that $x' < x$ and there is an edge $xx'$ in $G_{p^2(x)}$ then we label $x$ with $L$.
    \item Else, we label $x$ with $O$.
\end{itemize}

\begin{figure}[ht]
    \centering
    \begin{tikzpicture}[scale=0.5,auto=left,every node/.style={circle,draw,inner sep=2.6pt}]

		\node (1) at (-5,10) {\tiny{$\emptyset$}};
		\node (2) at (-9,8) {\tiny{$\emptyset$}};
		\node (3) at (-1,8) {\tiny{$\emptyset$}};

		\node (4) at (-9,5) {\tiny{$R$}};
		\node (5) at (-6,5) {\tiny{$L$}};

        \node (6) at (-1,5) {\tiny{$O$}};
		\node (7) at (4,5) {\tiny{$L$}};

		\node (8) at (-6,2) {\tiny{$R$}};
		\node (9) at (-3,2) {\tiny{$R$}};

        \node (10) at (-1,2) {\tiny{$O$}};
		\node (11) at (1,2) {\tiny{$R$}};
		\node (12) at (3,2) {\tiny{$L$}};
		\node (13) at (5,2) {\tiny{$L$}};

		\node (14) at (-3,0) {\tiny{$R$}};
		\node (15) at (-1.6,0) {\tiny{$L$}};
		\node (16) at (-0.4,0) {\tiny{$R$}};
		\node (17) at (1,0) {\tiny{$L$}};
		\node (18) at (2.4,0) {\tiny{$O$}};
		\node (19) at (3.6,0) {\tiny{$O$}};
		\node (20) at (5,0) {\tiny{$O$}};

		\node[fill=white] (21) at (-1.6,-2) {\tiny{$R$}};
		\node[fill=white] (22) at (-0.4,-2) {\tiny{$L$}};
		\node[fill=white] (23) at (2.4,-2) {\tiny{$O$}};
		\node[fill=white] (24) at (3.6,-2) {\tiny{$O$}};

		\draw(1) -- (2);
		\draw(1) -- (3);
		\draw(2) -- (4);
		\draw(3) -- (5);
		\draw(3) -- (6);
		\draw(3) -- (7);
		
		\draw(5) -- (8);
		\draw(6) -- (9);
		\draw(6) -- (10);
		\draw(6) -- (11);
		\draw(7) -- (12);
		\draw(7) -- (13);
		\draw(9) -- (14);
		\draw(10) -- (15);
		\draw(10) -- (16);
		\draw(11) -- (17);
		\draw(12) -- (18);
		\draw(12) -- (19);	
		\draw(13) -- (20);
		
		\draw(15) -- (21);
		\draw(16) -- (22);
		\draw(18) -- (23);
		\draw(19) -- (24);

		\draw[blue] (4) to [bend left=30] (5);
		\draw[blue] (4) to [bend left=30] (7);
		\draw[orange] (8) to[bend left=50] (12);
		\draw[orange] (9) to[bend left=50] (13);
		\draw[orange] (11) to[bend left=50] (12);
		\draw[orange] (11) to[bend left=50] (13);
		\draw[cyan] (14) to[bend left=50] (15);
		\draw[cyan] (16) to[bend left=50] (17);
		\draw[yellow] (21) to[bend left=50] (22);
        
    \end{tikzpicture}
    \caption{Labeling of a delayed structured tree.}
    \label{fig:delayed_labeling}
\end{figure}

\begin{lemma}\label{lem:no-consec-O}
Let $(T, <, \{G_x\}_{x \in V(T)})$ be a distinguishing delayed decomposition of an ordered graph $(G, <)$. If $x$ is a node of $T$ with at least $3$ children, there are no consecutive children $y_1, y_2$ of $x$ both labelled with $O$.
\end{lemma}

\begin{proof}
    Let $x$ be such a node, and $y_1, y_2$ be consecutive children of $x$.
    Since the delayed decomposition $(T, <, \{G_x\}_{x \in V(T)})$ is distinguishing, there is some vertex $v \in V(G) \setminus L(x)$ which is adjacent to all vertices in one of $L(y_1), L(y_2)$ and none in the other.
    Let $u_1$ be any vertex in $L(y_1)$ and $u_2$ any vertex in $L(y_2)$.
    Since $v \notin L(x)$ then $u_1$ and $v$ have the same last common ancestor $z$ as $u_2$ and $v$, and $z$ is a proper ancestor of $x$.
    Suppose by contradiction that $z$ is a proper ancestor of $p(x)$. 
    Then, the grandchild of $z$ which is an ancestor of $u_1$ is also the grandchild of $z$ which is an ancestor of $u_2$, call it $x'$. Let $v'$ be the ancestor of $v$ which is a grandchild of $z$.
    If $x'v' \in E(G_z)$ then $u_1v \in E(G)$ and $u_2v \in E(G)$, a contradiction. 
    Similarly, if $x'v' \notin E(G_z)$ then $u_1v \notin E(G)$ and $u_2v \notin E(G)$, again a contradiction. 
    Therefore, $z$ is the parent of $x$.
    Let $v'$ be be the ancestor of $v$ which is a grandchild of $z$. Then, $v'$ is a cousin of $y_1$ and $y_2$.
    If $i \in \{1, 2\}$ is such that $v$ is adjacent to all vertices in $L(y_i)$ then $v'y_i \in E(G_{p^2(y_i)})$ so one of $y_1$ and $y_2$ is not labelled with $O$.
\end{proof}

In view of \cref{lem:no-consec-O}, if $x$ is a node of $T$ with at least 3 children, and $y$ is a child of $x$ labelled $O$ which is not the first child of $x$, then the predecessor $y'$ of $y$ has either label $L$ or label $R$. If $y'$ has label $L$, we refine the label of $y$ to $O_L$ and if $y'$ has label $R$, we refine the label of $y$ to $O_R$.

Given a node $x$ of a delayed structured tree, if $y$ is a vertex of $G_x$ whose parent is labelled with $R$ then there exists some vertex $y' > V(G_x)$ which is adjacent to $y$. This observation will be our main tool to prove that graphs with large delayed rank have large complete interval minors (see \cref{thm:big-cano-rk-big-kt}). For this reason, we define the \emph{type} of a node $x \in V(T)$ as the label of its \emph{parent} in $T$.

The \emph{delayed rank} of an ordered graph $(G, <)$ is defined as follows: \begin{itemize}
    \item If $(G, <)$ is monotone bipartite, $(G, <)$ has delayed rank $0$,
    \item Otherwise, we compute the distinguishing delayed decomposition of $(G, <)$. For each quotient graph $(G_x, <)$, if $(G_x, <)$ is monotone bipartite we say that $(G_x, <)$ is a \emph{refined quotient graph} of $(G, <)$. Otherwise, note that $x$ has at least $3$ children (since there are only edges between cousins in $G_x$). In that case, if the first child of $x$ is labelled with $O$, we remove all its children from $G_x$. Then, we partition the vertices into types $R, L, O_R$ and $O_L$.
    Set $R' = \{R, O_R\}$ and $L' = \{L, O_L\}$.
    We partition the edges into four types, depending on whether the type of their left endpoint is in $R'$ or $L'$, and similarly for their right endpoint, denote these four types by $R'R', R'L', L'R'$ and $L'L'$.
    The \emph{refined quotient graphs} of $(G_x, <)$ are then defined as follows: \begin{itemize}
        \item The graph induced by the edges $R'R'$, to which we remove all vertices before the first vertex of type $R$ and after the last vertex of type $R$.
        \item The graph induced by the edges $R'L'$, to which we remove all vertices before the first vertex of type $L$ and after the last vertex of type $L$. 
        \item The graph induced by the edges $L'R'$, to which we remove all vertices before the first vertex of type $R$ and after the last vertex of type $R$.
        \item The graph induced by the edges $L'L'$, to which we remove all vertices before the first vertex of type $L$ and after the last vertex of type $L$.
    \end{itemize}
    Then, the delayed rank of $(G, <)$ is $1$ more than the maximum delayed rank of all its refined quotient graphs.
\end{itemize}

Observe first that siblings form an independent set in $G_x$, so removing all the children of the first child of $x$ is just removing an independent set, at the beginning of the order of $G_x$.
Similarly, the vertices before the first vertex of type $R$ are at the beginning of the order of $G_x$ and all have type either $L$ or $O_L$, so they cannot induce any edge of type $R'R'$ or $L'R'$.
As for the vertices after the last vertex of type $R$, they are at the end of the order of $G_x$ and all have type either $L, O_L$ or $O_R$, and in the latter case they are siblings and come before the vertices of type $L$ and $O_L$. Thus, all these vertices cannot induce any edge of type $R'R'$ or $L'R'$.
Similar observations hold for the vertices before the first vertex of type $L$ and after the last vertex of type $L$.

We now give some simple properties of the delayed rank. The first one gives a simple way of computing the delayed rank of an ordered graph. We first need some definitions.

If $G$ is an ordered graph, we define a sequence $(\mathcal{G}_i(G))_{i \geq 0}$ of sets of graphs as follows.
Set $\mathcal{G}_0(G) = \{G\}$. Suppose that $\mathcal{G}_i(G)$ is defined. 
Then, for every $H \in \mathcal{G}_i(G)$, add all the refined quotient graphs of $H$ to $\mathcal{G}_{i+1}(G)$.

\begin{lemma}\label{lem:large-rank-Gr}
An ordered graph $G$ has delayed rank at least $r$ if and only if ${\mathcal{G}_r(G) \neq \emptyset}$.
\end{lemma}

\begin{proof}
    We prove it by induction on $r$.
    The property is trivial for $r = 0$: every ordered graph $G$ has delayed rank at least $0$ and $\mathcal{G}_0(G) = \{G\} \neq \emptyset$.
    Suppose now that the property holds for some $r \geq 0$.
    The key observation is that for $s \geq 1$ we have $\mathcal{G}_{s}(G) = \bigcup_{H \in \mathcal{G}_1(G)}\mathcal{G}_{s-1}(H)$.
    Thus, using the induction hypothesis, we get: $G$ has delayed rank at least $r+1$ $\iff$ there is a refined quotient $H$ of $G$ (i.e. $H \in \mathcal{G}_1(G)$) of delayed rank at least $r$ $\iff$ there is a refined quotient $H$ of $G$ such that $\mathcal{G}_{r}(H) \neq \emptyset$ $\iff$ $\mathcal{G}_{r+1}(G) \neq \emptyset$.
\end{proof}

\begin{lemma}\label{lem:Gr-subgraph}
If $G, H$ are ordered graphs such that $H \in \mathcal{G}_r(G)$ for some $r \geq 0$ then $H$ is a subgraph of $G$.
\end{lemma}

\begin{proof}
    We prove it by induction on $r$.
    If $H \in \mathcal{G}_0(G)$ then $H = G$ so $H$ is a subgraph of $G$.
    Suppose that the property holds for $r$ and that $H \in \mathcal{G}_{r+1}(G)$.
    By definition, there exists $H' \in \mathcal{G}_r(G)$ such that $H$ is a refined quotient graph of $H'$.
    By induction hypothesis, $H'$ is a subgraph of $G$.
    Thus, to conclude it suffices to prove that $H$ is a subgraph of $H'$.
    To obtain $H$, we start from some quotient graph $G_x$ in the delayed decomposition of $H'$ (which is a subgraph of $H'$, as can be seen by taking a vertex in $L(t)$ for every $t \in V(G_x)$), and then possibly remove some vertices, take a subset of the edges and remove some more vertices. Thus, $H$ is a subgraph of $H'$.
\end{proof}

The following result bounds the size of each $\mathcal{G}_r(G)$. It will be important in the analysis of the running time of the algorithm in \cref{thm:the-algo}.

\begin{lemma} \label{lem:bound-size}
If $G$ is an ordered graph with $n$ vertices and $m$ edges, for every $r \geq 1$, the total number of edges of all graphs in $\mathcal{G}_r(G)$ is at most $m$.
Thus, for every $r \geq 1$, there there are at most $4mn$ graphs in $\mathcal{G}_r(G)$.
\end{lemma}

\begin{proof}
    We prove the first property by induction on $r$.
    For $r = 1$, \cref{rmk:nb-quotients} implies that the total number of edges of all quotient graphs of $G$ is at most $m$.
    Since the refined quotient graphs of $G$ are obtained from the quotient graphs by partitioning the edges and removing some edges, the total number of edges of all graphs in $\mathcal{G}_1(G)$ is at most $m$.
    Suppose now that the property holds for some $r \geq 1$.
    By definition, $\mathcal{G}_{r+1}(G) = \bigcup_{H \in \mathcal{G}_r(G)}\mathcal{G}_{1}(H)$.
    By the induction hypothesis at rank $1$, if $H$ has $m'$ edges then the total number of edges of all graphs in $\mathcal{G}_1(H)$ is at most $m'$.
    By the induction hypothesis at rank $r$, the total number of edges over all $H \in \mathcal{G}_r(G)$ is at most $m$. 
    Thus, the total number of edges of all graphs in $\mathcal{G}_{r+1}(G)$ is at most $m$.
    This proves the first part of the statement.
    
    If $m = 0$ then $G$ is monotone bipartite so $\mathcal{G}_1(G) = \emptyset$ and $|\mathcal{G}_1(G)| \leq 4mn$. Otherwise, by \cref{rmk:nb-quotients}, $G$ has at most $n-1$ quotient graphs, each of which gives rise to at most $4$ refined quotient graphs, so $|\mathcal{G}_{1}(G)| \leq 4(n-1) \leq 4mn$. 
    Now, suppose $r > 1$. By the first part of the statement, there are at most $m$ graphs $H \in \mathcal{G}_{r-1}(G)$ with at least one edge.
    All other graphs in $\mathcal{G}_{r-1}(G)$ are monotone bipartite, hence they have no refined quotient graphs.
    If $H \in \mathcal{G}_{r-1}(G)$ then $H$ has at most $n$ vertices by \cref{lem:Gr-subgraph} so, by \cref{rmk:nb-quotients}, $H$ has at most $n$ quotient graphs.
    Each of these quotient graphs gives rise to at most $4$ refined quotient graphs, so $|\mathcal{G}_{1}(H)| \leq 4n$. 
    Taking the union over all $H \in \mathcal{G}_{r-1}(G)$ which have at least one edge, we get $|\mathcal{G}_r(G)| \leq 4mn$.
\end{proof}

The next result shows that the notion delayed rank is simply a refinement of the notion of rank.

\begin{lemma}\label{lem:bddrk-bdrk}
If $G$ has delayed rank $r$ then $G$ has rank at most $7r+2$.
\end{lemma}

\begin{proof}
    We prove it by induction on $r$.
    If $G$ has delayed rank $0$ then $G$ is monotone bipartite so $G$ has rank at most $2$.
    Suppose now that $G$ has delayed rank $r$ and that the property holds for rank $r-1$.
    Compute the delayed decomposition of $G$. 
    By \cref{lem:delayed=subst+edge}, $G$ can be obtained from the quotient graphs by doing one step of substitution closure, followed by one step of edge union.
    Thus, it suffices to prove that the quotient graphs have rank at most $7r$.
    Each quotient graph can be obtained from the graphs with edges $R'R', R'L', L'R'$ and $L'R'$ by doing two steps of edge unions followed by the addition of a stable set at the beginning of the order (accounting for the possible removal of the first vertices of type $O$). 
    Thus, it suffices to prove that each of the graphs with edges $R'R', R'L', L'R'$ and $L'R'$ has rank at most $7(r-1) + 4$.
    By definition of the delayed rank, the graph induced by the edges $R'R'$ can be obtained from a refined quotient graph, of delayed rank at most $r-1$, by adding all vertices before the first vertex of type $R$ (which form a stable set), and all vertices after the last vertex of type $R$ (which also form a stable set). Thus, using the induction hypothesis, the graph induced by the edges $R'R'$ has rank at most $7(r-1) + 4$.
    Similarly, the graphs induced by the edges $R'L', L'R'$ and $L'L'$ each have rank at most $7(r-1) + 4$.
    This proves that $G$ has rank at most $7r+2$.
\end{proof}

\subsection{Delayed rank and complete interval minors}

We now prove the key result about graphs of large delayed rank: they contain large complete interval minors. 
The next lemma is the engine of our proof, it allows us to measure the progress we do in building the large complete interval minor when going from delayed rank $r$ to delayed rank $r+1$.

We consider \emph{looped} interval minors: an ordered graph $(H, <)$ (possibly with a loop on each vertex) with vertex set $v_1 < \ldots < v_h$ is an interval minor of an ordered graph $(G, <)$ if there exists a partition of $V(G)$ into intervals $I_1, \ldots, I_h$ such that whenever $v_iv_j \in E(H)$, there is an edge in $G$ between $I_i$ and $I_j$.
A \emph{looped $K_t$} is an ordered clique of size $t$, with a loop on every vertex.
A \emph{left-lazy looped $K_t$} is an ordered clique of size $t$, with a loop on every vertex, except the first one.
A \emph{right-lazy looped $K_t$} is an ordered clique of size $t$, with a loop on every vertex, except the last one.
A \emph{lazy looped $K_t$} is an ordered clique of size $t$, with a loop on every vertex, except the first one and the last one.
See~\cref{figKT} for an illustration of all these graphs.

\begin{figure}[!ht]
	\centering
		\begin{tikzpicture}[every node/.style={circle, draw, minimum size=4mm}, scale=1]

			\foreach \i in {1,2,3,4} {
				\node (A\i) at (\i*2, 0) {
				};
			}
			\foreach \a/\b in {1/2,1/3,1/4,2/3,2/4,3/4} {
				\draw[bend left=20] (A\a) to (A\b);
			}
			\foreach \i in {1,2,3,4} {
				\draw[-] (A\i) to[loop below] (A\i);
			}
			\node[draw=none] at (5, -1.2) {Looped $K_t$};

			\foreach \i in {1,2,3,4} {
				\node (B\i) at (\i*2, -3) {};
			}
			\foreach \a/\b in {1/2,1/3,1/4,2/3,2/4,3/4} {
				\draw[bend left=20] (B\a) to (B\b);
			}
			\foreach \i in {2,3,4} {
				\draw[-] (B\i) to[loop below] (B\i);
			}
			\node[draw=none] at (5, -4.2) {Left-lazy Looped $K_t$};

			\foreach \i in {1,2,3,4} {
				\node (C\i) at (\i*2, -6) {};
			}
			\foreach \a/\b in {1/2,1/3,1/4,2/3,2/4,3/4} {
				\draw[bend left=20] (C\a) to (C\b);
			}
			\foreach \i in {1,2,3} {
				\draw[-] (C\i) to[loop below] (C\i);
			}
			\node[draw=none] at (5, -7.2) {Right-lazy Looped $K_t$};

			\foreach \i in {1,2,3,4} {
				\node (D\i) at (\i*2, -9) {};
			}
			\foreach \a/\b in {1/2,1/3,1/4,2/3,2/4,3/4} {
				\draw[bend left=20] (D\a) to (D\b);
			}
			\foreach \i in {2,3} {
				\draw[-] (D\i) to[loop below] (D\i);
			}
			\node[draw=none] at (5, -10.2) {Lazy Looped $K_t$};

			\draw[very thick, -stealth] (8.5, -0.5) to [out=0, in=0] (8.5, -2.5);
			\draw[very thick, -stealth] (8.5, -0.5) to [out=0, in=0] (8.5, -5.5);
			\draw[very thick, -stealth] (8.5, -3.5) to [out=0, in=0] (8.5, -8.5);
			\draw[very thick, -stealth] (8.5, -6.5) to [out=0, in=0] (8.5, -8.5);
			
			\draw[very thick, -stealth] (1.5, -2.5) to [out=180, in=180] (1.5, -0.5);
			\draw[very thick, -stealth] (1.5, -5.5) to [out=180, in=180] (1.5, -0.5);
			\draw[very thick, -stealth] (1.5, -8.5) to [out=180, in=180] (1.5, -3.5);
			\draw[very thick, -stealth] (1.5, -8.5) to [out=180, in=180] (1.5, -6.5);

			\node[draw=none] at (9, -7.5) {\textbf{\small {Size +1}}};
			\node[draw=none] at (9.9, -6) {\textbf{\small {Size +1}}};
			\node[draw=none] at (9, -1.5) {\textbf{\small {Size +1}}};
			\node[draw=none] at (9.9, -3) {\textbf{\small {Size +1}}};
			
		\end{tikzpicture}
	\caption{The various looped interval minors we consider. The arrows indicate the possible outcomes of \cref{lem:increase-Kt-delayed}.}
	\label{figKT}
\end{figure}
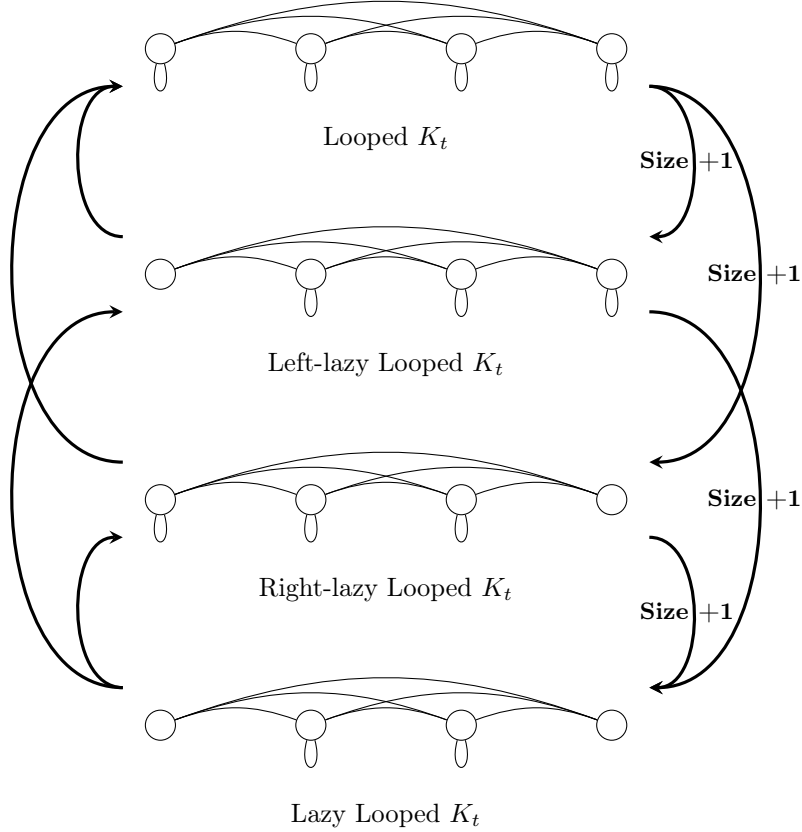

\begin{lemma}\label{lem:increase-Kt-delayed}
Let $r \geq 1$ and let $\mathcal{C}_r$ be the class of ordered graphs with delayed rank $r$. Then, the following assertions hold.
\begin{enumerate}
    \item If every ordered graph in $\mathcal{C}_r$ contains a looped $K_t$ interval minor, then every ordered graph in $\mathcal{C}_{r+1}$ contains a left-lazy or a right-lazy looped $K_{t+1}$ interval minor. \label{item:a}
    \item If every ordered graph in $\mathcal{C}_r$ contains a left-lazy or a right-lazy looped $K_t$ interval minor, then every ordered graph in $\mathcal{C}_{r+1}$ contains either a lazy looped $K_{t+1}$ interval minor or a looped $K_t$ interval minor. \label{item:b}
    \item If every ordered graph in $\mathcal{C}_r$ contains a lazy looped $K_t$ interval minor, then every ordered graph in $\mathcal{C}_{r+1}$ contains a left-lazy or a right-lazy looped $K_t$ interval minor. \label{item:c}
\end{enumerate}
\end{lemma}

\begin{proof}
    Consider an ordered graph $(G,<)\in \mathcal{C}_{r+1}$. Compute the delayed decomposition of $(G, <)$. By definition, there is a refined quotient graph $H$ of $G$ some type ($R'R', R'L', L'R'$ or $L'L'$) which is in $\mathcal{C}_r$.
    We consider several cases depending on the type of $H$.
    \begin{itemize}
        \item If $H$ has type $R'R'$ then $H$ is the graph induced by the edges $R'R'$, to which we remove all vertices before the first vertex of type $R$ and after the last vertex of type $R$.
        Suppose that $H$ contains a lazy looped $K_t$ interval minor, and let $I_1 < \ldots < I_t$ be the intervals corresponding to this interval minor. Recall that $(I_1, \ldots, I_t)$ is a partition of $V(H)$.
        We argue that each interval $I_j$ contains a vertex of type $R$.
        
        Since the first and the last vertex of $H$ are vertices of type $R$, this is true for $I_1$ and $I_t$.
        Consider now any $I_j$ with $1 < j < t$.
        Since the $K_t$ interval minor is looped, there is an edge between two vertices of $I_j$, which means that there are two vertices in $I_j$ which are of type $R'$ but don't have the same parent (since siblings form an independent set in the quotient graphs). Therefore, there is a vertex of type $R$ in $I_j$.
        Let $I_{t+1} = \{x \in V(G) : V(H) < x\}$. If $y_j \in I_j$ is a vertex of type $R$, it follows from the definition of the type $R$ that $y_j$ has a neighbor in $I_{t+1}$.
        
        Using $I_{t+1}$, we can then extend any looped $K_t$ to a right lazy looped $K_{t+1}$, any left-lazy looped $K_t$ to a lazy looped $K_{t+1}$, any right-lazy looped $K_{t}$ to a looped $K_t$ and any lazy looped $K_{t}$ to a right-lazy looped $K_t$.
        
        \item If $H$ has type $R'L'$ then $H$ is the graph induced by the edges $R'L'$, to which we remove all vertices before the first vertex of type $L$ and after the last vertex of type $L$.
        Suppose that $H$ contains a lazy looped $K_t$ interval minor, and let $I_1 < \ldots < I_t$ be the intervals corresponding to this interval minor. Recall that $(I_1, \ldots, I_t)$ is a partition of $V(H)$.
        We argue that each interval $I_j$ contains a vertex of type $L$.
        
        Since the first and the last vertex of $H$ are vertices of type $L$, this is true for $I_1$ and $I_t$.
        Consider now any $I_j$ with $1 < j < t$.
        Since the $K_t$ interval minor is looped, there is an edge between two vertices of $I_j$, which means that there exist $u < v \in I_j$ such that the type of $u$ is in $R' = \{R, O_R\}$, and the type of $v$ is in $L' = \{L, O_L\}$. 
        Thus, there exist consecutive vertices $u < v \in I_j$ such that the type of $u$ is in $R'$ and the type of $v$ is in $L'$. 
        This implies that the type of $v$ is $L$. 
        Therefore, there is a vertex with type $L$ in $I_j$.
        Let $I_{0} = \{x \in V(G) : x < V(H)\}$. If $y_j \in I_j$ is a vertex of type $L$, it follows from the definition of the type $L$ that $y_j$ has a neighbor in $I_{0}$.
        
        Using $I_{0}$, we can extend any looped $K_t$ to a left lazy looped $K_{t+1}$, any left-lazy looped $K_t$ to a looped $K_{t}$, any right-lazy looped $K_{t}$ to a lazy looped $K_{t+1}$ and any lazy looped $K_{t}$ to a left-lazy looped $K_t$.
        
        \item The case $L'R'$ is similar to the case $R'L'$ (except that we prove that each $I_j$ contains a vertex of type $R$), and the case $L'L'$ is similar to the case $R'R'$ (except that we prove that each $I_j$ contains a vertex of type $L$).
    \end{itemize}
    
    The result then follows immediately since we can extend any lazy looped $K_t$ interval minor of $H$ as desired, and since $H \in \mathcal{C}_r$.
\end{proof}

We easily deduce the main result of this section from \cref{lem:increase-Kt-delayed}.

\begin{theorem}\label{thm:big-cano-rk-big-kt}
Every ordered graph with delayed rank at least $3r-2$ contains a $K_r$ interval minor.
\end{theorem}

\begin{proof}
    We prove by induction on $r$ that every ordered graph with delayed rank at least $3r-2$ has a looped $K_{r}$ interval minor.
    For $r = 1$, if $G$ has delayed rank at least $1$ then $G$ contains at least one edge so $G$ has a looped $K_1$ interval minor.
    Suppose that the property holds for $3r-2$.
    By \cref{lem:increase-Kt-delayed}, every graph of delayed rank at least $3r-1$ contains a left-lazy or a right-lazy looped $K_{t+1}$ interval minor.
    Applying \cref{lem:increase-Kt-delayed} again, every graph of delayed rank at least $3r$ contains either a lazy looped $K_{t+2}$ interval minor or a looped $K_{t+1}$ interval minor.
    Using \cref{lem:increase-Kt-delayed} once more, every graph of delayed rank at least $3r+1$ contains either a left-lazy or a right-lazy looped $K_{t+2}$ interval minor, which itself contains a looped $K_{t+1}$ interval minor.
\end{proof}

Now, \cref{thm:boundedrank} follows immediately from combining \cref{thm:big-cano-rk-big-kt} with \cref{lem:bddrk-bdrk}.

%% file: Analysis_algo.tex
\section{Approximating the complete interval minor number} \label{sec:algo}

In this section, we present an algorithm to approximate the size of a largest complete interval minor in an ordered graph $G$. 
In \cref{subsec:algo-tools}, we give some implementation details on parts of the algorithm. 
In \cref{subsec:ramsey}, we give a Ramsey-type theorem in the context of interval minors, which is crucial for bounding the approximation factor of the algorithm.
In \cref{subsec:the-algo}, we describe and analyse the algorithm.
Finally, in \cref{subsec:algo-K3}, we give a $O(n)$-time algorithm to decide whether an $n$-vertex ordered graph contains a $K_3$ interval minor.

\subsection{More algorithmic tools}\label{subsec:algo-tools}

By \cref{thm:compute-delayed}, the delayed decomposition of an explicit ordered graph with $n$ vertices and $m$ edges can be computed in time $O(n+m)$. From there, it is simple to compute the refined quotients in the same running time.

\begin{lemma}\label{lem:compute-refined-quotients}
There is an algorithm which, given as input an explicit ordered graph $G$ with $n$ vertices and $m$ edges, computes all the refined quotients of $G$ in time $O(n + m)$.
\end{lemma}

\begin{proof}
First, we start by checking whether $G$ is monotone bipartite. 
To do so, we iterate over all edges to find the largest left endpoint of an edge, call it $\ell$, and the smallest right endpoint of an edge, call it $r$. If $\ell < r$ then $G$ is monotone bipartite and has no refined quotients, so we return the empty set. Otherwise, $G$ is not monotone bipartite.

In that case, we compute the delayed decomposition $(T, <, \{G_x\}_{x \in V(T)})$ of $G$ in time $O(n+m)$ using \cref{thm:compute-delayed}, with all the $G_x$ stored explicitly. 
Then, we compute the label of every node $x$ ($L, R$ or $O$) by looking at its neighborhood in the graph $G_{p^2(x)}$.
Thus, in time $O(n+m)$, we can compute the label of all nodes $x \in V(T)$.
From this, we can compute the type ($L, R, O_L, O_R$ or $O$) of every node of $T$, in time $O(n+m)$ overall.
Then, for each quotient graph $G_x$, we check whether it is monotone bipartite. If so, we add it to the set of refined quotient graphs. Otherwise, we continue. Using the same method as above, this can be done in time linear in the size of $G_x$, so in time $O(n+m)$ overall.

We then remove from each $G_x$ the first vertex as long as it is of type $O$.
Then, for every edge of every $G_x$, we can compute its type by looking at the types of its endpoints. 
Once this is done, it is simple to compute the graphs induced by each of the four edge types. 
Finally, removing all vertices up to the first vertex of some type and after the last vertex of some type can also be done in time $O(n + m)$ overall. Note that the refined quotients are all stored explicitly.
\end{proof}

The \emph{total size} of a delayed structured tree $(T, <, \{G_x\}_{x \in V(T)})$ is ${|V(T)| + \sum_{x \in V(T)}|E(G_x)|}$.

\begin{lemma}\label{lem:algo-heavy-leaf}
There is an algorithm which, given as input a delayed structured tree ${(T, <, \{G_x\}_{x \in V(T)})}$ of total size $s$, computes in time $O(s)$ whether there is a $h$-heavy leaf in $T$.
\end{lemma}

\begin{proof}
By iterating over all the edges of all the $G_x$, we can mark in time $O(s)$ all the nodes $x \in V(T)$ which are the endpoint of an edge (in which case it will be in $G_{p^2(x)}$).
Then, a simple top-down dynamic programming algorithm can compute for every node $x \in V(T)$ the number $a(x)$ of ancestors $x'$ of $x$ which are not isolated in $G_{p^2(x')}$, in total time $O(s)$.
Finally, by iterating over all the leaves of $T$, it is easy to compute the maximum value of $a(x)$ over all leaves $x$ of $T$, and return Yes if this maximum is at least $h$, and No otherwise.
\end{proof}

\subsection{A Ramsey-type result for complete interval minors}\label{subsec:ramsey}

Before we move on to the algorithm, we present a Ramsey-type result for complete interval minors, which will be crucial for the bound on the ``approximation factor'' of our algorithm. Ramsey's theorem states that in every red/blue coloring of the edges of $K_n$, there is a monochromatic clique of size $\log(n)/2$, and this bound is essentially tight up to constant factors.
We consider the analog question in the context of interval minors for ordered graphs.
A red/blue coloring of the edges of the ordered graph $K_n$ contains a \emph{red complete interval minor} of size $t$ if there exists a partition of $V(G)$ into $t$ intervals such that there is a red edge between any two of these intervals. We define similarly a \emph{blue complete interval minor} of size $t$, and a \emph{monochromatic complete interval minor} of size $t$ is either a red complete interval minor of size $t$ or a blue complete interval minor of size $t$.

We now prove that the bounds are much better for monochromatic complete interval minors than for monochromatic cliques.

\begin{lemma}\label{lem:improved-ramsey}
For every red/blue coloring of the edges of the ordered graph $K_n$, there is a monochromatic complete interval minor of size $2^{\sqrt{\log(n)}-1}$.
\end{lemma}

\begin{proof}
    Suppose that there does not exist a red complete interval minor of size $2^{\sqrt{\log(n)}-1}$.
    We prove that for every $i \leq \sqrt{\log(n)}-1$, we can find $2^i$ intervals, each of size at least $n/2^{i\sqrt{\log(n)}}$, with only blue edges between any two of these intervals.
    The property is trivial for $i = 0$.
    Suppose that the property holds for some $i < \sqrt{\log(n)}-1$.
    Then, there exist $2^i$ intervals, each of size at least $n/2^{i\sqrt{\log(n)}}$, with only blue edges between any two of them.
    Consider one such interval $I$, and cut it into $2^{\sqrt{\log(n)}-1}$ subintervals, each of size at least $\lfloor |I|/2^{\sqrt{\log(n)}-1}\rfloor$.
    Then, each subinterval has size at least ${\lfloor|I|/2^{\sqrt{\log(n)}-1}\rfloor \geq n/2^{(i+1)\sqrt{\log(n)}-1} - 1 \geq n/2^{(i+1)\sqrt{\log(n)}}}$.
    Indeed: \begin{align*}
        \frac{n}{2^{(i+1)\sqrt{\log(n)}-1}} - 1 \geq \frac{n}{2^{(i+1)\sqrt{\log(n)}}} &\iff \frac{2n-2^{(i+1)\sqrt{\log(n)}}}{2^{(i+1)\sqrt{\log(n)}}} \geq \frac{n}{2^{(i+1)\sqrt{\log(n)}}} \\
        &\iff 2n-2^{(i+1)\sqrt{\log(n)}} \geq n \\
        &\iff n \geq 2^{(i+1)\sqrt{\log(n)}} \\
        &\iff \log(n) \geq (i+1)\sqrt{\log(n)} \\
        &\iff i \leq \sqrt{\log(n)}-1.
    \end{align*}
    Since there does not exist a red complete interval minor of size $2^{\sqrt{\log(n)}-1}$, there are two subintervals with no red edge between them, hence only blue edges between them.
    Doing this in each of the $2^i$ intervals, we find $2^{i+1}$ subintervals, each of size at least $n/2^{(i+1)\sqrt{\log(n)}}$, with only blue edges between any two of them.
    This result for $i = \sqrt{\log(n)}-1$ proves the existence of a blue complete interval minor of size $2^{\sqrt{\log(n)}-1}$.
\end{proof}

The next result shows that the previous bound is almost sharp.

\begin{lemma}
    For $n$ large enough, there exists a red/blue coloring of the edges of the ordered graph $K_n$ for which the largest monochromatic complete interval minor has size $2^{2 \cdot \sqrt{\log(n) \cdot \log\log(n)}}$.
\end{lemma}

\begin{proof}
    Set $q = 2^{\sqrt{\log(n) \cdot \log\log(n)}}$. 
    Consider a red/blue edge coloring of the edges of the ordered graph $K_q$ with no monochromatic clique or independent set of size $3\log(q)$ (a random coloring satisfies this property with high probability if $q$ is large enough).
    Then, consider the red/blue coloring of the edges of the ordered graph $K_{q^2}$ obtained by substituting every vertex of $K_q$ by a copy of the ordered clique $K_q$ with the previous coloring (the order on the vertices of $K_{q^2}$ can then be seen as the lexicographic order on the vertices of $K_q$).
    Repeat this process to obtain a red/blue coloring of the edges of the ordered graphs $K_{q^3}, K_{q^4}$, and so on until $K_{q^k} = K_n$.
    Note that $$q^k = n \iff k \cdot \sqrt{\log(n) \cdot \log\log(n)} = \log(n) \iff k = \sqrt{\log(n) / \log\log(n)}.$$
    
    For every $i \in [k]$, let $f(i)$ be the size of a largest monochromatic complete interval minor in $K_{q^i}$.
    First, observe that $f(1) \leq q$.
    Then, consider $i > 1$.
    Denote by $v_1, \ldots, v_q$ the vertices of $K_q$. Observe that $K_{q^i}$ can be obtained from $K_q$ by substituting each vertex by a copy of $K_{q^{i-1}}$.
    Let $\mathcal{I}$ be the set of intervals in a monochromatic complete interval minor of $K_{q^{i}}$ of size $f(i)$.
    Let $v_{i_1}, \ldots, v_{i_{\ell}}$ be the set of vertices $v$ of $K_q$ such that some interval of $\mathcal{I}$ is entirely contained into the copy of $K_{q^{i-1}}$ which was substituted for $v$.
    Then, $v_{i_1}, \ldots, v_{i_{\ell}}$ induce a monochromatic clique in the original $K_q$, so $\ell \leq 3\log(q)$.
    For each of them, the restriction of $\mathcal{I}$ to the intervals entirely inside the corresponding copy of $K_{q^{i-1}}$ contains at most $f(i-1)$ intervals.
    Furthermore, for every vertex $v$ of $K_q$, there is at most one interval of $\mathcal{I}$ which starts in the copy of $K_{q^{i-1}}$ which was substituted for $v$, and which does not end in that copy.
    Thus, there are at most $q$ intervals of $\mathcal{I}$ which are not entirely contained inside a copy of $K_{q^{i-1}}$.
    Overall, this yields $f(i) = |\mathcal{I}| \leq 3\log(q) \cdot f(i-1) + q$.
    A straightforward induction then yields $f(i) \leq q \cdot \sum_{j=0}^{i-1}(3\log(q))^j$, which in turn implies $f(i) \leq q \cdot (3\log(q))^i$.
    
    For $i=k=\sqrt{\log(n) / \log\log(n)}$, using that $n$ is large, we get that the largest monochromatic complete interval minor in this red/blue coloring of the edges of $K_n$ has size at most 
    \begin{align*}
        f(k) &\leq q \cdot (3\log(q))^k \\
            &= 2^{\sqrt{\log(n) \cdot \log\log(n)}} \cdot \left(3\sqrt{\log(n) \cdot \log\log(n)}\right)^{\sqrt{\log(n) / \log\log(n)}} \\
            &\leq 2^{\sqrt{\log(n) \cdot \log\log(n)}} \cdot \log(n)^{\sqrt{\log(n) / \log\log(n)}} \\
            &= 2^{\sqrt{\log(n) \cdot \log\log(n)}} \cdot 2^{\log\log(n) \cdot \sqrt{\log(n) / \log\log(n)}} \\
            &= 2^{2 \cdot \sqrt{\log(n) \cdot \log\log(n)}}. \qedhere
    \end{align*}
\end{proof}

It would be interesting to determine more precisely the largest function $f(n)$ such that every every red/blue coloring of the edges of the ordered graph $K_n$ contains a monochromatic complete interval minor of size $f(n)$. The results of this section show that ${2^{\sqrt{\log(n)}-1} \leq f(n) \leq 2^{2 \cdot \sqrt{\log(n) \cdot \log\log(n)}}}$.

\subsection{The algorithm} \label{subsec:the-algo}

We now move to the description of the algorithm.
We will need the following technical lemma to bound the ``approximation factor'' of our algorithm. Its proof is deferred to \cref{app:proof_technical}.

\begin{restatable}{lemma}{technical} \label{lem:technical}
For every $t \geq 1$, there exists a function $h_t : \{0, 1, \ldots, 3t-2\} \to \mathbb{R}^+$, such that the following holds: 
\begin{enumerate}
    \item $f: t \mapsto h_t(0)$ satisfies $f(t) = 2^{2^{2^{O(t)}}}$, 
    \item For every $r \in [3t-2]$, $h_t(r) \leq 2^{\sqrt{\sqrt{\log\left(\left(h_t(r-1)/2\right)^{1/(2t-1)}-4\right)}-1}-1}-4$,
    \item For every $r \in \{0, \ldots, 3t-2\}$, $h_t(r) \geq 4$.
\end{enumerate}
\end{restatable}

We are now ready to prove \cref{thm:the-algo}. In the following statement, we assume that the ordered graph $(G, <)$ is given explicitly.

\thealgo*

\begin{proof}
The high-level description of the algorithm is extremely simple: we compute the delayed rank of $(G, <)$. If this rank is at least $3t-2$, we return Yes. Otherwise, we look whether there is a $(2t-3)$-heavy leaf in one of the delayed structured trees that we computed. If so, we return Yes. Otherwise, we return No.

More formally, we compute the sets $\mathcal{G}_0(G), \ldots, \mathcal{G}_{3t-2}(G)$.
If $\mathcal{G}_{3t-2}(G) \neq \emptyset$, we return Yes. Otherwise, if there exists some $H \in \mathcal{G}_0(G) \cup \ldots \cup \mathcal{G}_{3t-2}(G)$ whose delayed decomposition tree contains a $(2t-3)$-heavy leaf then we return Yes.
Otherwise, we return No.

For every $t \geq 1$, let $h_t$ be the function provided by \cref{lem:technical}. 
Then, define $f : t \mapsto h_t(0)$, and note that $f(t) = 2^{2^{2^{O(t)}}}$.

\begin{claim*}
If the algorithm returns Yes then $(G, <)$ has a $K_t$ interval minor.
\end{claim*}

\begin{proof}[\textit{Proof of the Claim}]
By \cref{lem:large-rank-Gr}, if $\mathcal{G}_{3t-2}(G) \neq \emptyset$ then $G$ has delayed rank at least $3t-2$ so \cref{thm:big-cano-rk-big-kt} implies that $G$ has a $K_t$ interval minor.
If there exists some $H \in \mathcal{G}_0(G) \cup \ldots \cup \mathcal{G}_{3t-2}(G)$ whose delayed decomposition tree contains a $(2t-3)$-heavy leaf then by \cref{lem:Gr-subgraph}, $H$ is a subgraph of $G$ and $H$ is the realization of a delayed structured tree that contains a $(2t-3)$-heavy leaf so by \cref{lem:cliquepath2}, $H$ contains a clique of size $t$. Therefore, $G$ itself contains a clique of size $t$, hence a $K_t$ interval minor.
\end{proof}

\begin{claim*}
If $G$ has a $K_{f(t)}$ interval minor then the algorithm returns Yes.
\end{claim*}

\begin{proof}[\textit{Proof of the Claim}]
Suppose that there is no graph $H \in \mathcal{G}_0(G) \cup \ldots \cup \mathcal{G}_{3t-2}(G)$ whose delayed decomposition tree contains a $(2t-3)$-heavy leaf (otherwise we return Yes).
We prove by induction that for every $0 \leq r \leq 3t-2$, there exists a graph $H_r \in \mathcal{G}_r(G)$ which has a complete interval minor of size $h_t(r)$. In particular, this will imply that $\mathcal{G}_{3t-2}(G) \neq \emptyset$, so the algorithm returns Yes.
For $r = 0$, set $H_0 = G \in \mathcal{G}_0(G)$, which has a complete interval minor of size $f(t) = h_t(0)$ by assumption.
Suppose that the property holds for some $r \in \{0, 1, \ldots, 3t-3\}$.
Consider an interval family $\mathcal{I}$ which realizes the complete interval minor of size $h_t(r)$ in $H_r$. 
In particular, since $h_t(r) \geq 4$ then $H_r$ is not bipartite so the refined quotient graphs of $H_r$ are in $\mathcal{G}_{r+1}(G)$.
Let $(T, <, \{G_x\}_{x \in V(T)})$ be the delayed decomposition of $H_r$.
Set $b_r = (h_t(r)/2)^{1/(2t-1)} -2$, so that $h_t(r) = 2(b_r+2)^{2t-1}$.
By \cref{lem:cliquepath}, since $H_r$ doesn't contain a $(2t-3)-$heavy leaf, there is no $(2t-1)$-interval path in $H_r$. Then, \cref{lem:branchtree} implies that there is a $b_r$-branching node $x$ in $T$, which means that the corresponding quotient graph $G_x$ has a complete interval minor of size $b_r$.
After possibly removing the first vertices of type $O$, the resulting subgraph of $G_x$ still has a complete interval minor of size $b_r-2$ since we only removed an independent set.
Applying \cref{lem:improved-ramsey} twice, we get that one of the types $R'R', R'L', L'R'$ and $L'R'$ satisfies that the subgraph induced by the edges of this type has a complete interval minor of size at least $2^{\sqrt{\sqrt{\log(b_r-2)}-1}-1}$.
Removing the first vertices and the last vertices, which both form a stable set, decreases the size of the complete interval minor by at most $4$, so one of the refined quotient graphs $H_{r+1}$ of $H_r$ has a complete interval minor of size at least $2^{\sqrt{\sqrt{\log(b_r-2)}-1}-1} - 4 = 2^{\sqrt{\sqrt{\log\left((h_t(r)/2)^{1/2t-1} -4\right)}-1}-1} - 4 \geq h_{t}(r+1)$ by definition of $h_t$.
Thus, $H_{r+1} \in \mathcal{G}_{r+1}(G)$ has a complete interval minor of size at least $h_t(r+1)$, as desired.
\end{proof}

\begin{claim*}
This algorithm can be implemented to run in time $O(t \cdot mn^2)$.
\end{claim*}

\noindent\textit{Proof of the Claim.}
By \cref{lem:compute-refined-quotients}, the set of refined quotient graphs (stored explicitly) of an explicit ordered graph with $v$ vertices and $e$ edges can be computed in time $O(v+e)$.
Thus, $\mathcal{G}_1(G)$ can be computed in time $O(n+m)$, with all graphs being stored explicitly.

Suppose that we already computed $\mathcal{G}_r(G)$ for some $1 \leq r < 3t-2$, with all graphs being stored explicitly. Then, \cref{lem:Gr-subgraph,lem:bound-size} imply that it contains at most $4mn$ graphs, each with at most $n$ vertices, and with at most $m$ edges in total over all graphs in $\mathcal{G}_r(G)$.
Thus, $\mathcal{G}_{r+1}(G)$ can be computed in time $\sum_{H \in \mathcal{G}_r(G)}O(|V(H)|+|E(H)|) = O(4mn \cdot n + m) = O(mn^2)$, with all graphs being stored explicitly.
Therefore, computing $\mathcal{G}_0(G) \cup \ldots \cup \mathcal{G}_{3t-2}(G)$ can be done in time $O(t \cdot mn^2)$.

Finally, given an explicit ordered graph with $v$ vertices and $e$ edges, its delayed decomposition can be computed in time $O(v+e)$, hence the corresponding delayed structured tree has total size $O(v+e)$. Then, by \cref{lem:algo-heavy-leaf}, the algorithm can compute in time $O(v+e)$ whether there is a $(2t-3)$-heavy leaf in it.
Thus, by iterating over all graphs in $\mathcal{G}_0(G) \cup \ldots \cup \mathcal{G}_{3t-2}(G)$, the algorithm can check whether there exists some $H \in \mathcal{G}_0(G) \cup \ldots \cup \mathcal{G}_{3t-2}(G)$ whose delayed structured tree contains a $(2t-3)$-heavy leaf.
Since by \cref{rmk:nb-quotients} there are $O(t \cdot mn)$ such graphs and together they contain at most $O(t \cdot m)$ edges, this can be done in time \begin{equation*}
    \sum_{H \in \mathcal{G}_0(G) \cup \ldots \cup \mathcal{G}_{3t-2}(G)}O(|V(H)| + |E(H)|) = O(t \cdot mn \cdot n + t \cdot m) = O(t \cdot mn^2). \qedhere \qedsymbol
\end{equation*}
\end{proof}

\begin{remark}
Observe that the proofs of \cref{lem:cliquepath2,lem:increase-Kt-delayed,thm:big-cano-rk-big-kt} are all algorithmic and can be implemented efficiently. Note also that in the course of the algorithm, we compute all the graphs in $\mathcal{G}_0(G) \cup \ldots \cup \mathcal{G}_{3r-2}(G)$ and their delayed decompositions. Therefore, when the algorithm returns Yes, it can also efficiently return a collection of intervals that form a $K_t$ interval minor in $(G, <)$.
\end{remark}

\subsection{Finding a \texorpdfstring{$K_3$}{K3} interval minor} \label{subsec:algo-K3}

Every non-trivial ordered graph contains $K_1$ as an interval minor, and every ordered graph with at least one edge contains $K_2$ as an interval minor. Therefore, deciding whether an ordered graph contains $K_1$ or $K_2$ as an interval minor can be done in constant time. In this section, we provide a linear-time algorithm for deciding whether an ordered graph contains $K_3$ as an interval minor.

\begin{theorem} \label{thm:algo-K3}
There is an algorithm which, given as input an explicit $n$-vertex ordered graph $(G, <)$, decides whether $K_3$ is an interval minor of $(G, <)$ in time $O(n)$.
\end{theorem}

The next lemma explains why the above running time can be as low as $O(n)$, instead of the usual $O(m+n)$. 
It contrasts strikingly with the fact that $K_4$-interval-minor-free ordered graphs can have a quadratic number of edges, as we observed in the introduction.

\begin{lemma} \label{lem:n-edges-K3}
If $(G, <) = ((V, E), <)$ is an ordered graph with $|E| \geq |V|$ then $(G, <)$ contains $K_3$ as an interval minor.
\end{lemma}

\begin{proof}
Since $|E| \geq |V|$ then $G$ contains a cycle $(v_1, \ldots, v_\ell, v_1)$.
Without loss of generality, assume that $v_1$ is $<$-minimum among all $v_i$, and that $v_2 < v_\ell$.
Consider the intervals $I_1, I_2$ and $I_3$, defined respectively as: \begin{itemize}
    \item All vertices $u$ such that $u \leq v_1$, and
    \item All vertices $u$ such that $v_1 < u \leq v_2$, and
    \item All vertices $u$ such that $v_2 < u$.
\end{itemize}
Observe that $v_1 \in I_1, v_2 \in I_2$ and $v_{\ell} \in I_{3}$.
The edge $v_1v_2$ is an edge between $I_1$ and $I_2$, and the edge $v_1v_{\ell}$ is an edge between $I_1$ and $I_3$.
Furthermore, since $v_1$ is $<$-minimum among all $v_i$, the path $v_2, \ldots, v_{\ell}$ is a path between $I_2$ and $I_3$ which never visits $I_1$, so it must contain an edge between $I_2$ and $I_3$. 
Thus, the interval minor obtained by contracting $I_1, I_2$ and $I_3$ is $K_3$.
\end{proof}

\begin{proof}[Proof of \cref{thm:algo-K3}]
The algorithm is presented in \cref{alg:algo-K3}.
The high-level idea is to try all possible endpoints for the first interval, and then determine whether the remainder of the graph can be partitioned into two intervals satisfying the required conditions, using constant-time queries.

\begin{algorithm}
\caption{Algorithm for \cref{thm:algo-K3}}
        \label{alg:algo-K3}
        \begin{algorithmic}[1]
        \Require $(G, <) = ((V, E), <)$ on vertex set $[n]$ given by its edge set.
    
        \If{$G$ has at least $n$ edges} \label{line:count-edges}
            \State Return Yes \label{line:first-return}
        \EndIf
        \For{every vertex $v \in V$}
            \State $M(v) \gets$ largest neighbor of $v$ ($0$ if $v$ is isolated) \label{line:computeM}
            \State $Y(v) \gets$ $1$ if $M(v) > v$, $0$ otherwise
            \State $m(v) \gets$ smallest $u > v$ which has a neighbour $w \leq v$ ($\infty$ if no such $u$) \label{line:computem}
        \EndFor
        \State Preprocess the arrays $M$ and $Y$ \label{line:preprocess}
        \For{every vertex $v \in V$}
            \State $I_1 \gets \{x \in V : x \leq v\}$
            \State $L \gets m(v)$
            \State $R \gets \max_{x \in I_1}M(x)$
            \If{there exists $u \in V$ such that $L \leq u < R$ and $M(u) > u$}
                \State Return Yes \label{line:second-return}
            \EndIf
            \If{there exists $u \in V$ such that $v < u < L < R$ and $M(u) > L$}
                \State Return Yes \label{line:last-return}
            \EndIf
        \EndFor
        \State Return No
        \end{algorithmic}
\end{algorithm}

\begin{claim*}
If $(G, <)$ has a $K_3$ interval minor then the algorithm returns Yes.
\end{claim*}

\begin{proof}[\noindent\textit{Proof of the Claim}]
If $G$ has at least $n$ edges, the algorithm returns Yes at \cref{line:first-return} so suppose that $G$ has less than $n$ edges. 
Let $(I_1, I_2, I_3)$ be a partition of $V(G)$ into intervals for $<$ which realizes $K_3$ as an interval minor, with $I_1 < I_2 < I_3$, and let $v$ be the last vertex of $I_1$.
Let $L$ be the minimum vertex outside of $I_1$ which has a neighbor in $I_1$, and $R$ be the maximum vertex which has a neighbor in $I_1$. Note that $L = m(v)$ and $R = \max_{u \in I_1}M(u)$. Since some vertex in $I_2$ has a neighbor in $I_1$ then $L \in I_2$, and similarly $R \in I_3$. This implies $L < R$.
Furthermore, there is an edge between $I_2$ and $I_3$, let $u \in I_2$ be a vertex with a neighbor in $I_3$.
If $u \geq L$ then $L \leq u < R$ (since $R \in I_3)$ and $M(u) > u$ since the neighbor $w$ of $u$ in $I_3$ satisfies $M(u) \geq w > u$. Thus, in that case the algorithm returns Yes at \cref{line:second-return}.
If $u < L$ then $v < u < L < R$ (since $v$ is the last vertex of $I_1$) and $M(u) > L$ since the neighbor $w$ of $u$ in $I_3$ satisfies $M(u) \geq w > L$ (since $L \in I_2$). Thus, in that case the algorithm returns Yes at \cref{line:last-return}.
Therefore, in all cases the algorithm returns Yes.
\end{proof}

\begin{claim*}
If the algorithm returns Yes then $(G, <)$ has a $K_3$ interval minor.
\end{claim*}

\begin{proof}[\noindent\textit{Proof of the Claim}]
Suppose first that the algorithm returns Yes at \cref{line:first-return}. Then, $G$ has at least $n$ edges so \cref{lem:n-edges-K3} implies that $(G, <)$ has a $K_3$ interval minor.

Suppose now that the algorithm returns Yes at \cref{line:second-return} while considering some vertex $v \in V$. Then, there exists $u \in V$ such that $L \leq u < R$ and $M(u) > u$. In particular, $L \neq \infty$ and $R \neq 0$. Let $I_1 = \{x \in V : x \leq v\}$, $I_2 = \{x \in V : v < x \leq u\}$ and $I_3 = \{x \in V : x > u\}$.
Since $L = m(v)$ and $L \neq \infty$ then by definition $L > v$ and $L$ has a neighbor $x \leq v$, i.e. $L$ has a neighbor in $I_1$. Since $v < L \leq u$ then $L \in I_2$ so there is an edge between $I_1$ and $I_2$.
Since $R \neq 0$ then there is an edge between $R$ and some vertex $x \leq v$, i.e. $R$ has a neighbor in $I_1$. Since $u < R$ then $R \in I_3$ so there is an edge between $I_1$ and $I_3$. 
Furthermore, since $M(u) > u$ then $M(u) \neq 0$ so $u$ has a neighbor $w > u$, so there is an edge between $I_2$ and $I_3$. 
This proves that the partition $(I_1, I_2, I_3)$ realizes $K_3$ as an interval minor.

Last, suppose that the algorithm returns Yes at \cref{line:last-return} while considering some vertex $v \in V$. Then, there exists $u \in V$ such that $v < u < L < R$ and $M(u) > L$. In particular, $L \neq \infty$ and $R \neq 0$. Let $I_1 = \{x \in V : x \leq v\}$, $I_2 = \{x \in V : v < x \leq L\}$ and $I_3 = \{x \in V : x > L\}$.
Since $L = m(v)$ and $L \neq \infty$ then by definition $L$ has a neighbor $x \leq v$, i.e. $L$ has a neighbor in $I_1$. Thus, there is an edge between $I_1$ and $I_2$. 
Since $R \neq 0$ then there is an edge between $R$ and some vertex $x \leq v$, i.e. $R$ has a neighbor in $I_1$. Since $L < R$ then $R \in I_3$ so there is an edge between $I_1$ and $I_3$.
Furthermore, since $v < u < L$ then $u \in I_2$, and since $M(u) > L$ then $u$ has a neighbor $w > L$, i.e. there is an edge between $I_2$ and $I_3$.
This proves that the partition $(I_1, I_2, I_3)$ realizes $K_3$ as an interval minor.
\end{proof}

\begin{claim*}
This algorithm can be implemented to run in time $O(n)$.
\end{claim*}

\noindent\textit{Proof of the Claim. }
Since $(G, <)$ is given explicitly, we can relabel all vertices in time $O(n)$ so that $G$ is on vertex set $[n]$.
\cref{line:count-edges} can be executed in time $O(n)$.
If the algorithm doesn't return Yes at \cref{line:first-return} then $G$ has less than $n$ edges so computing all $M(v)$ for $v \in V$ can be done in time $O(n)$ by iterating over all edges.
Computing the array $Y$ can then be done in time $O(n)$.
Similarly, computing for every $v \in V$ the smallest neighbor $s(v)$ of $v$ can be done in time $O(n)$ by iterating over all edges.
We now show how to compute all $m(v)$ in time $O(n)$.
Initialize an empty stack $S$, and consider all vertices of $G$ in turn, from $1$ to $n$.
When considering a vertex $v$, let $t$ be the vertex on top of $S$. While $s(v) \leq t$, set $m(t) = v$, and remove $t$ from $S$. When this no longer holds, add $v$ on top of $S$ and consider the next vertex.
After having added the vertex $n$ to $S$, set $m(v) = \infty$ for all $v$ that remain in $S$.
The correction of this algorithm can be proved with the following invariant: When the algorithm is considering a vertex $v$, all vertices $u \in S$ satisfy $m(u) \geq v$.

The preprocessing at \cref{line:preprocess} takes time $O(n)$ by \cref{lem:rmq}, and it allows us to answer maximum and minimum queries in constant time on any interval in the arrays $M$ and $Y$.

Then, the algorithm iterates over all $v \in V$. We show that for every $v \in V$, only constant time is spent considering $v$.
Computing $R$ takes constant time because it is a query on an interval for $M$, and $L$ can simply be accessed as $m(v)$.
Checking whether there exists $u \in V$ such that $L \leq u < R$ and $M(u) > u$ can be done by querying the maximum of the array $Y$ in the interval $[L, R-1]$ and checking whether this maximum is $1$ or $0$, which takes constant time using \cref{lem:rmq}.
Checking whether there exists $u \in V$ such that $v < u < L < R$ and $M(u) > L$ can be done by querying the maximum of the array $M$ in the interval $[v+1, L-1]$ and checking whether this maximum is greater than $L$, which also takes constant time using \cref{lem:rmq}. \qedhere \qedsymbol
\end{proof}

%% file: appendix.tex
\section{Proof of \texorpdfstring{\cref{lem:technical}}{Lemma 26}} \label{app:proof_technical}

\technical*

\begin{proof}
For every integer $t\geq 1$ and every integer $r \in \{0, \ldots, 3t-2\}$, set $${g_t(r) = 4^{3t-2-r} + \frac{4^{3t-2-r}-1}{3}\log(512t)}$$ and ${h_t(r) = 2^{2^{g(r)}}}$. 

Note that $g_t$ is nonincreasing, and ${g_t(3t-2) = 1}$, so for every ${r \in \{0, \ldots, 3t-2\}}$ we have ${h_t(r) \geq 2^{2^1} = 4}$. Moreover, it is clear by definition that ${f : t \mapsto h_t(0)}$ satisfies ${f(t) = 2^{2^{2^{O(t)}}}}$.

Fix some integer $t \geq 1$. To simplify the notation, we denote the function $h_t$ by $h$.
We will check that for every ${r \in [3t-2]}$, we indeed have 
$${h(r) \leq 2^{\sqrt{\sqrt{\log\left(\left(h(r-1)/2\right)^{1/(2t-1)}-4\right)}-1}-1}-4}.$$

Since ${h(r) + 4 \leq 2h(r)}$, it suffices to prove that 
\begin{align*}
        & 4h(r) \leq 2^{\sqrt{\sqrt{\log\left(\left(h(r-1)/2\right)^{1/(2t-1)}-4\right)}-1}}\\ 
        \iff & \log(4h(r)) \leq \sqrt{\sqrt{\log\left(\left(h(r-1)/2\right)^{1/(2t-1)}-4\right)}-1}\\
        \iff & \log(4h(r))^2 \leq \sqrt{\log\left(\left(h(r-1)/2\right)^{1/(2t-1)}-4\right)}-1 .
\end{align*}

Since ${\log(4h(r))^2 +1 \leq (\log(4h(r))+1)^2 = \log(8h(r))^2}$, it suffices to prove that 
\begin{align*}
        & \log(8h(r))^2 \leq \sqrt{\log\left(\left(h(r-1)/2\right)^{1/(2t-1)}-4\right)}\\ 
        \iff & \log(8h(r))^4 \leq \log\left(\left(h(r-1)/2\right)^{1/(2t-1)}-4\right)\\
        \iff & 2^{\log(8h(r))^4} \leq \left(h(r-1)/2\right)^{1/(2t-1)}-4 .
\end{align*}

Since ${2^{\log(8h(r))^4} + 4 \leq 2^{\log(8h(r))^4+1} \leq 2^{\log(16h(r))^4}}$, it suffices to prove that 

\begin{align*}
        & 2^{\log(16h(r))^4} \leq \left(h(r-1)/2\right)^{1/(2t-1)}\\ 
        \iff & \log(16h(r))^4 \leq \frac{1}{2t-1} \cdot \log(h(r-1)/2)\\
        \iff & (2t-1)\log(16h(r))^4 \leq \log(h(r-1)) -1 .
\end{align*}

Note that $$(2t-1)\log(16h(r))^4 + 1 \leq 2t\log(16h(r))^4 = 2t(\log(h(r)) + 4)^4.$$ Since $h(r) \geq 4$ we have $\log(h(r)) + 4 \leq 4\log(h(r))$. So $$(2t-1)\log(16h(r))^4 + 1 \leq 2t(4\log(h(r)))^4 \leq 512t\log(h(r))^4.$$
Thus, it suffices to show that $512t \log(h(r))^4 \leq \log(h(r-1))$.
Since $g_t(r) = \log\log(h(r))$, it suffices to show $4g_t(r) + \log(512t) \leq g_t(r-1)$.

It is simple from the definition of $g_t$ to see that $4g_t(r) + \log(512t) = g_t(r-1)$.
This proves that the desired inequality holds.
\end{proof}